%% file: main.tex
\documentclass[sigconf]{acmart}
\AtBeginDocument{%
  \providecommand\BibTeX{{%
    \normalfont B\kern-0.5em{\scshape i\kern-0.25em b}\kern-0.8em\TeX}}}

\usepackage{lipsum}
\usepackage{verbatim}
\usepackage{pifont}
\usepackage{xcolor} 
\usepackage[super]{nth}
\usepackage{comment}
\usepackage{graphicx}
\usepackage{tabularx}
\usepackage{amsmath}
\usepackage{bm}
\usepackage{pgfplots}
\usepackage{caption}
\usepackage{subcaption}
\usepackage{kotex}
\usepackage{graphics} 
\usepackage{adjustbox}
\usepackage{multirow}
\usepackage{float}
\usepackage{url}
\PassOptionsToPackage{hyphens}{url}\usepackage{hyperref}
\usepackage[linesnumbered,ruled,boxed,vlined]{algorithm2e}
\usepackage{paralist}
\usepackage{placeins}
\usepackage{natbib}
\usepackage{mathtools}

\copyrightyear{2020}
\acmYear{2020}
\setcopyright{acmcopyright}\acmConference[KDD '20]{Proceedings of the 26th ACM SIGKDD Conference on Knowledge Discovery and Data Mining}{August 23--27, 2020}{Virtual Event, USA}
\acmBooktitle{Proceedings of the 26th ACM SIGKDD Conference on Knowledge Discovery and Data Mining (KDD '20), August 23--27, 2020, Virtual Event, USA}
\acmPrice{15.00}
\acmDOI{10.1145/3394486.3403060}
\acmISBN{978-1-4503-7998-4/20/08}

\input{dfn}

\newcommand\blue[1]{\textcolor{black}{#1}}
\newcommand\kijung[1]{\textcolor{red}{[Kijung: #1]}}
\definecolor{myGreen2}{rgb}{0.133,0.543,0.133}

\newcommand\red[1]{\textcolor{red}{#1}}

\begin{document}

\title{Structural Patterns and Generative Models of Real-world Hypergraphs}


\author{Manh Tuan Do}
\affiliation{%
\institution{KAIST EE}
}
\email{manh.it97@kaist.ac.kr}

\author{Se-eun Yoon}
\affiliation{%
\institution{KAIST EE}	
}
\email{granelle@kaist.ac.kr}

\author{Bryan Hooi}
\affiliation{%
  \institution{NUS School of Computing}
}
\email{bhooi@comp.nus.edu.sg}

\author{Kijung Shin}
\authornote{Corresponding author.}
\affiliation{%
\institution{KAIST AI \& EE}
}
\email{kijungs@kaist.ac.kr}

\renewcommand{\shortauthors}{Do et al.}

\begin{abstract}
\vspace{-1mm}
Graphs have been utilized as a powerful tool to model pairwise relationships between people or objects. Such structure is a special type of a broader concept referred to as hypergraph, in which each hyperedge may consist of an arbitrary number of nodes, rather than just two. A large number of real-world datasets are of this form – for example, lists of recipients of emails sent from an organization, users participating in a discussion thread or subject labels tagged in an online question. However, due to complex representations and lack of adequate tools, little attention has been paid to exploring the underlying patterns in these interactions.

In this work, we empirically study a number of real-world hypergraph datasets across various domains. In order to enable thorough investigations, we introduce the multi-level decomposition method, which represents each hypergraph by a set of pairwise graphs.
Each pairwise graph, which we refer to as a $k$-level decomposed graph, captures the interactions between pairs of subsets of $k$ nodes.
We empirically find that at each decomposition level, the investigated hypergraphs obey five structural properties. These properties serve as criteria for evaluating how realistic a hypergraph is, and establish a foundation for the hypergraph generation problem. We also propose a hypergraph generator that is remarkably simple but capable of fulfilling these evaluation metrics, which are hardly achieved by other baseline generator models.
\vspace{-1mm}
\end{abstract}
\maketitle


\vspace{-1mm}
\section{Introduction}
\vspace{-1mm}
\label{sec:intro}
\input{010introduction}

\vspace{-1mm}
\section{Background and Related Work}
\vspace{-1mm}
\label{sec:background}
\input{020background}

\vspace{-1mm}
\section{Multi-level Decomposition}
\vspace{-1mm}
\label{sec:meth}
\input{030projection}

\vspace{-1mm}
\section{Observations}
\vspace{-1mm}
\label{sec:exp}
\input{040observation}

\vspace{-1mm}
\section{Hypergraph Generators}
\vspace{-1mm}
\label{sec:gen}
\input{050generator}

\vspace{-1mm}
\section{Conclusions}
\vspace{-1mm}
\label{sec:concl}
\input{060conclusion}

\vspace{-1mm}
\subsection*{Acknowledgements}
\vspace{-1mm}
{\small This work was supported by National Research Foundation of Korea (NRF) grant funded by the
	Korea government (MSIT) (No. NRF-2020R1C1C1008296) and Institute of Information \& Communications
	Technology Planning \& Evaluation (IITP) grant funded by the Korea government (MSIT) (No. 2019-0-00075, Artificial Intelligence Graduate School Program (KAIST)).}

\bibliographystyle{ACM-Reference-Format}
\bibliography{references}
\clearpage

\appendix
\input{070appendix_reproducibility.tex}

\clearpage
\end{document}

%% file: dfn.tex
\newtheorem{theorem}{Theorem}

\newcommand{\hide}[1]{}

\newcommand{\bit}{\begin{compactitem}[$\bullet$]}
\newcommand{\eit}{\end{compactitem}}
\newcommand{\ben}{\begin{compactenum}}
\newcommand{\een}{\end{compactenum}}

\newcommand{\generator}{\textsc{HyperPA}\xspace}
\newcommand{\ssl}{Subset Sampling\xspace}
\newcommand{\floor}[1]{\left\lfloor #1 \right\rfloor}

\definecolor{myGreen}{rgb}{0.133,0.543,0.133}
\newcommand{\cmark}{\textcolor{myGreen}{\ding{51}}}%
\newcommand{\xmark}{\textcolor{red}{\ding{55}}}%

\newcommand{\nodeproj}{{node-level decomposed graphs}\xspace}
\newcommand{\edgeproj}{{edge-level decomposed graphs}\xspace}
\newcommand{\triangleproj}{{triangle-level decomposed graphs}\xspace}

\newcommand{\Edgeproj}{{Edge-level decomposed graphs}\xspace}
\newcommand{\Triangleproj}{{Triangle-level decomposed graphs}\xspace}
\newcommand{\Fcliqueproj}{{4clique-level decomposed graphs}\xspace}

\newcommand{\connectcomp}{{Giant connected component}\xspace}
\newcommand{\degdistr}{{Heavy-tailed degree distribution}\xspace}
\newcommand{\effdiam}{{Small effective diameter}\xspace}
\newcommand{\clustcoef}{{High clustering coefficient}\xspace}
\newcommand{\sngvals}{{Skewed singular values}\xspace}

\newcommand{\smallsection}[1]{{\vspace{0.05in} \noindent {\bf{\underline{\smash{#1}}}}}}

\newcommand{\Interaction}{Subset interaction\xspace}
\newcommand{\interaction}{subset interaction\xspace}
\newcommand{\interactions}{subset interactions\xspace}

%% file: 010introduction.tex
\begin{figure}[t]
	\centering
	\includegraphics[width=0.8\columnwidth]{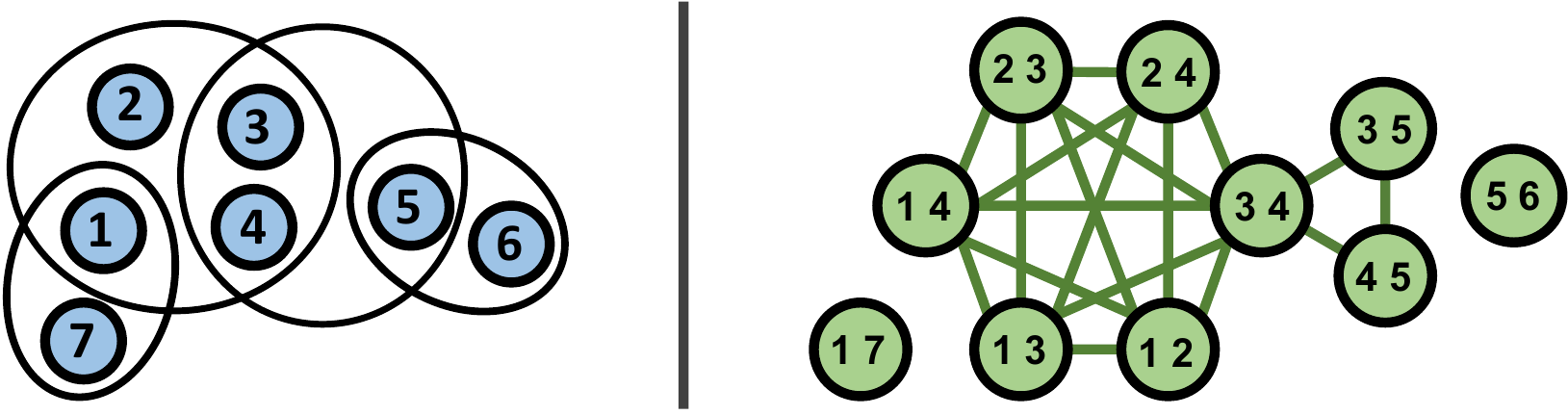} \\
	\vspace{-0.5mm}
	\caption{\label{fig:example}
		A hypergraph and its $2$-level decomposed graph.
	}
\end{figure}
\input{summarytags-ask-ubuntu-figure_shrink.tex}

In our digital age, interactions that involve a group of people or objects are ubiquitous \cite{benson2018simplicial, benson2018sequences, benson2016higher}. These associations 
arise from various domains, ranging from academic communities, online social networks to pharmaceutical practice. In particular, research papers are often published by the collaborations of several coauthors, social networks involve group communications, and several related medications may be applied as a treatment rather than just two.

Such structures can be represented as \textit{hypergraphs} \cite{berge1989hyper, bonacich2004hyper}, which is a generalization of the usual notion of graphs. In hypergraphs, each node can be a person or an object. However, each hyperedge acts as an interaction of an arbitrary number of nodes. For example, if each node represents an author, a hyperedge can be treated as a research paper which was published by a group of authors. 
A hyperedge also reveals the \textit{\interactions} among the elements of each subset, which this work pays special attention to.
A \interaction among nodes (e.g., $\{a,b\}$) is defined as their co-appearance as a subset of a hyperedge (e.g., $\{a,b,c,d\}$).  
The freedom of number of nodes involved in each hyperedge and \interactions naturally contribute to the complexity of hypergraphs.

While pairwise graphs have been extensively studied in terms of mining structures \cite{milo2002network, broder2000graph, girvan2002community}, discovering hidden characteristics \cite{faloutsos1999power, bollobas2004diameter, eikmeier2017revisiting, klein1999the} 
as well as evolutionary patterns \cite{paranjape2017motifs, leskovec2008microscopic, leskovec2005graphs}, little attention has been paid to defining and addressing analogous problems in hypergraphs. Due to the complexity of \interactions, any single representation of hypergraphs relying on pairwise links would suffer from information loss. Given that most existing graph data structures only capture relationships between pairs of nodes, and more importantly, most patterns discovered are based on pairwise links-based measurements, directly applying the existing results in pairwise graphs to hypergraphs constitutes a challenge.

Here we investigate several hypergraph datasets among various domains \cite{benson2018simplicial, sinha2015MAG, yin2017local}. We introduce the \textit{multi-level decomposition} of hypergraphs, which captures relationships between subsets of nodes. This offers a set of pairwise link representations convenient for analysis while guaranteeing to recover the original hypergraphs. In the most elementary type of decomposition, referred to as ``node-level decomposed graph'' in this paper, two nodes are linked if they appear in at least one hyperedge together. This is the decomposition for $k=1$. In the $k$-level decomposed graph, a node is defined as a set of $k$ nodes in the original hypergraph, and two nodes are connected if their union appears in a hyperedge (see Fig.~\ref{fig:example}).

Using the multi-level decomposition, we find that the decomposed graphs of thirteen real-world hypergraphs generally obey the following well-known properties of real-world graphs, across different levels: (1) {\it giant connected components}, (2) {\it heavy-tailed degree distributions}, (3) {\it small effective diameters}, (4) {\it high clustering coefficients}, and (5) {\it skewed singular-value distributions}.
This decomposition also reveals how well such \interactions are connected, and this connectivity varies across different domains.

What could be the possible underlying principles for such patterns?
Driven by this question, we propose a simple hypergraph generator model called \generator. By some proper modifications of \textit{preferential attachment} \cite{barabasi1999emergence, albert2002statistical, klein1999the}, which account for degree as a group, nodes can ``get rich'' together while maintaining \interactions. Compared to two other baseline models, \generator shows more realistic results in reproducing the patterns discovered in real-world hypergraphs and resembling the connectivity of such \interactions (see Fig.~\ref{fig:summary_tags-ask-ubuntu_figure}).

Findings in common properties of real-world hypergraphs and their underlying explanations can be significant for several reasons:
(1) {\textit{anomaly detection}}: if some data significantly deviates from the set of common patterns, it is reasonable to raise an alarm for anomalies,
(2) {\textit{anonymization}}: by fully reproducing these patterns, organizations may synthesize datasets to avoid disclosing important internal information. 
(3) {\textit{simulation}}: generated hypergraphs can be utilized for ``what-if'' simulation scenarios when collecting large-size hypergraph datasets is costly and difficult.

In short, the main contributions of our paper are three-fold.
\bit
\item {\bf Multi-level decomposition}: a tool that facilitates easy and comprehensive analysis of \interactions in hypergraphs.
\item {\bf Patterns}: five structural properties that are commonly held in thirteen real-world hypergraphs from diverse domains.
\item {\bf Hypergraph generator (\generator)}: a simple but powerful model that produces hypergraphs satisfying the above properties.
\eit
{\bf Reproducibility:} We made the datasets, the code, and the full experimental results available at \url{https://github.com/manhtuando97/KDD-20-Hypergraph}.

The remaining sections of this paper are outlined as follow: Sect. \ref{sec:background} provides a brief survey of related work. In Sect. \ref{sec:meth}, we introduce our \textit{decomposition} tool which facilitates our understanding of structural properties of hypergraphs. Our empirical findings on real-world hypergraph datasets are presented in Sect. \ref{sec:exp}. Sect. \ref{sec:gen} introduces hypergraph generators and demonstrates how these models perform in terms of reproducing the real-world patterns. We discuss and conclude our work in Sect. \ref{sec:concl}.

%% file: summarytags-ask-ubuntu-figure_shrink.tex
\begin{figure*}[htbp]
	\addtolength{\tabcolsep}{-3pt}
	\begin{center}
		\vspace{-4mm}
		\centering
		\scalebox{0.8}{
		\begin{tabular}{cc|cc|cc}
	    \toprule
		\multicolumn{2}{c|}{\bf Edge Level} & \multicolumn{2}{c|}{\bf Triangle Level} & \multicolumn{2}{c}{\bf 4clique Level} \\ 
		\midrule  
		Real Data & {\bf \generator (Proposed)} 
		& Real Data & {\bf \generator (Proposed)} 
		& Real Data & {\bf \generator (Proposed)} \\			
	    \midrule
		\includegraphics[width=1.4in]{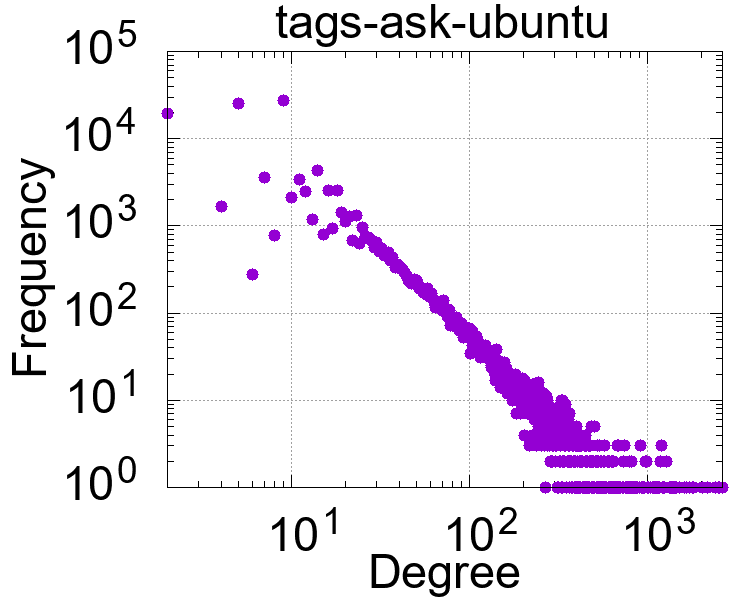} &
		\includegraphics[width=1.4in]{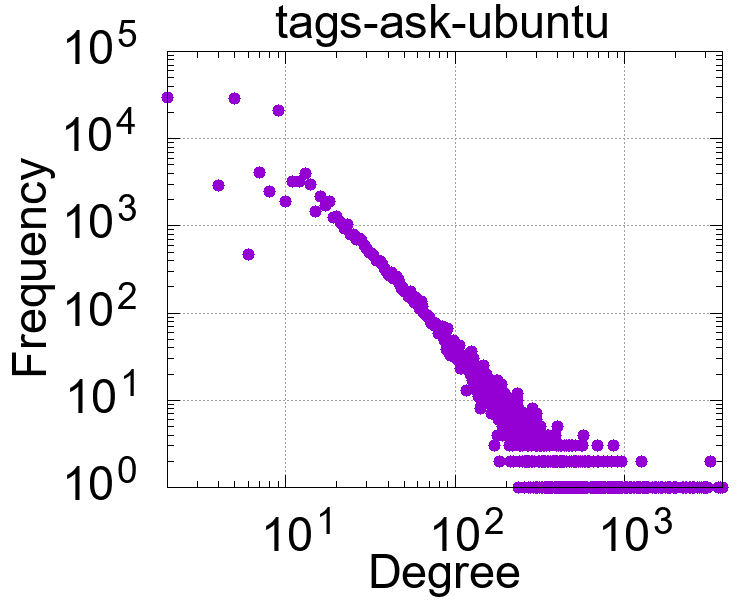} & 
		\includegraphics[width=1.4in]{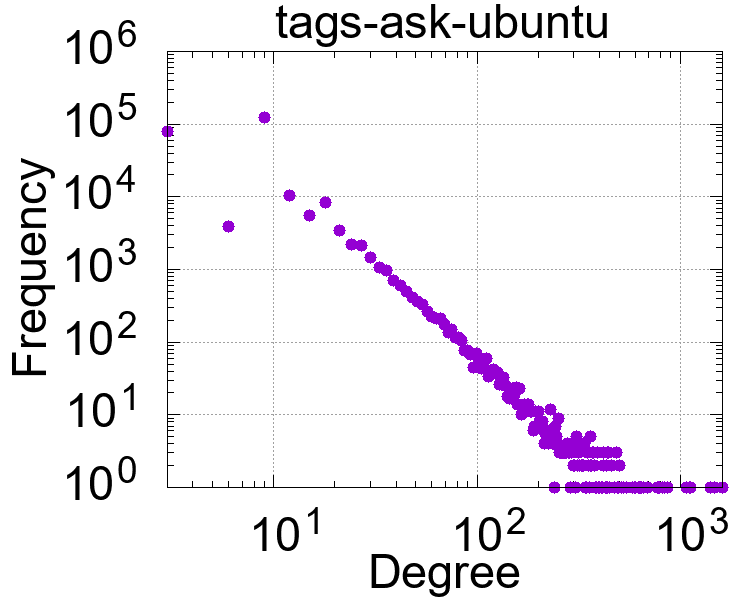} &
	    \includegraphics[width=1.4in]{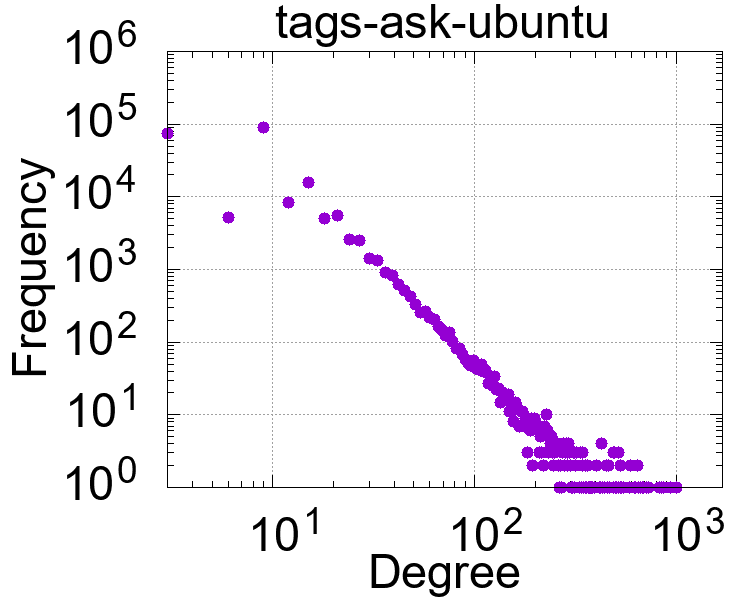} & 
	    \includegraphics[width=1.4in]{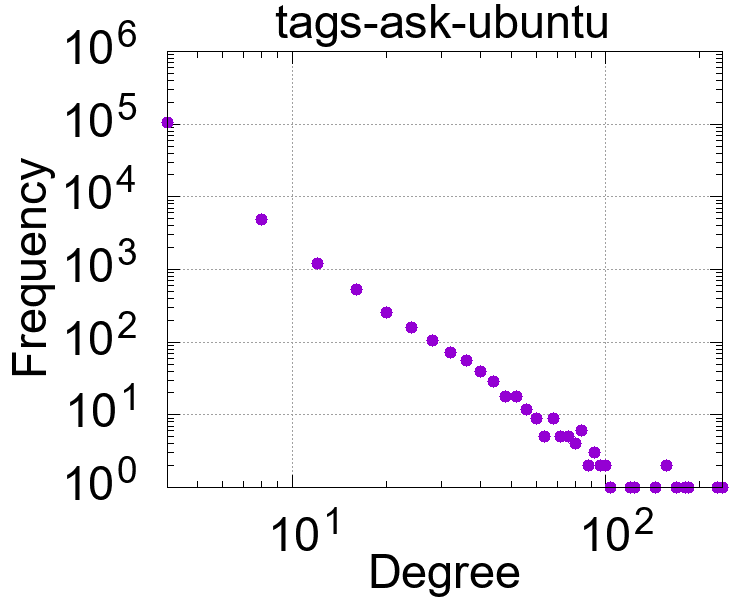} &
		\includegraphics[width=1.4in]{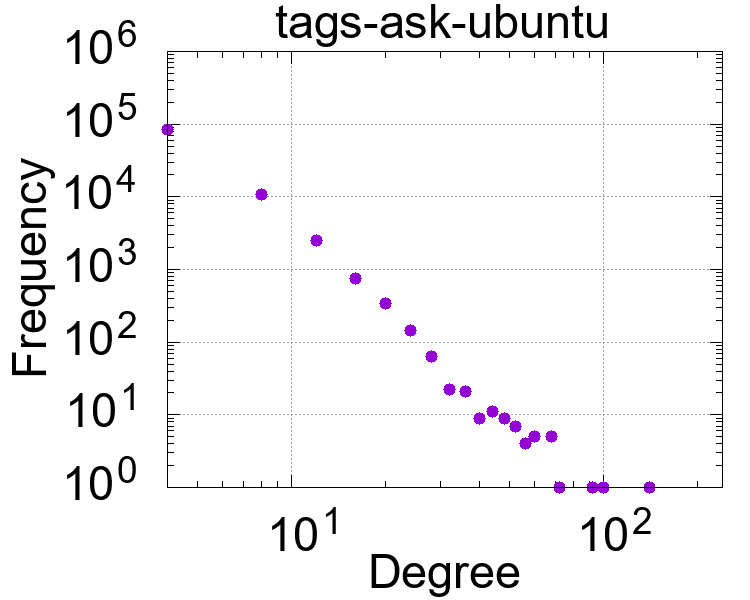} \\
		\midrule
		NaivePA & SubsetSamplling & NaivePA & SubsetSampling & NaivePA & SubsetSampling \\
		\midrule
		\includegraphics[width=1.4in]{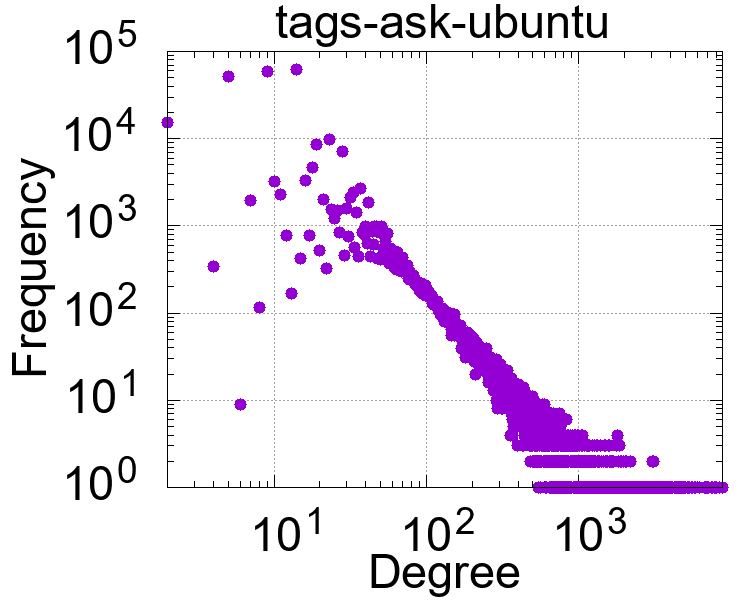} &
		\includegraphics[width=1.4in]{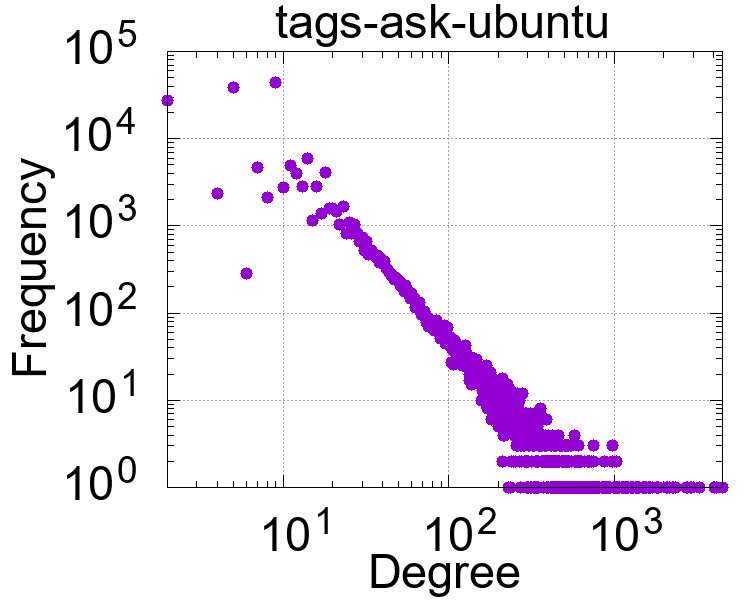} &	
	\includegraphics[width=1.4in]{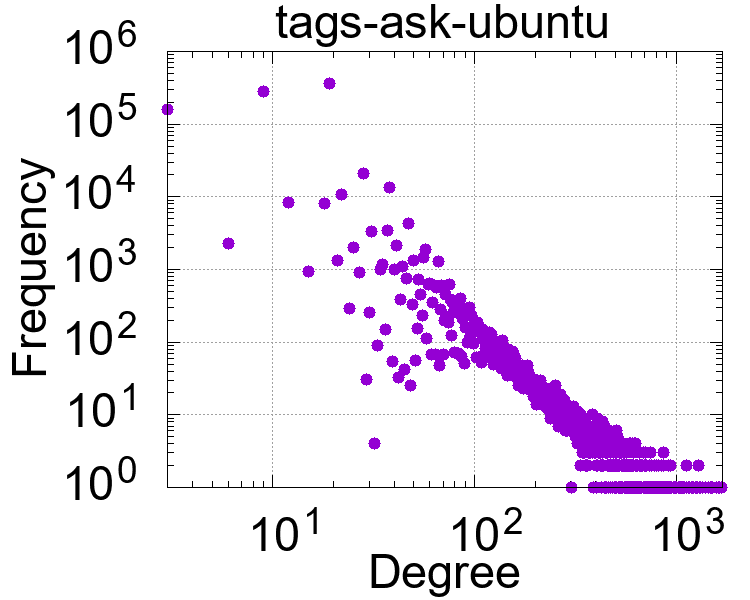} &
	\includegraphics[width=1.4in]{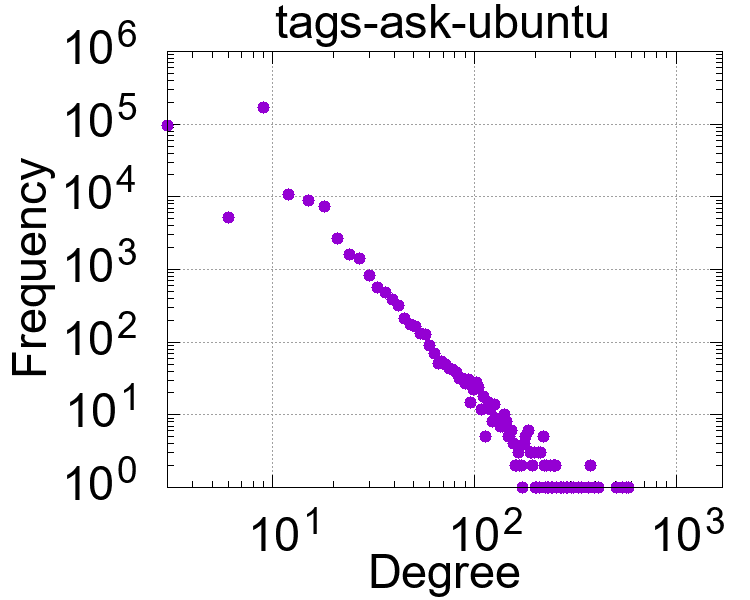} &
		\includegraphics[width=1.4in]{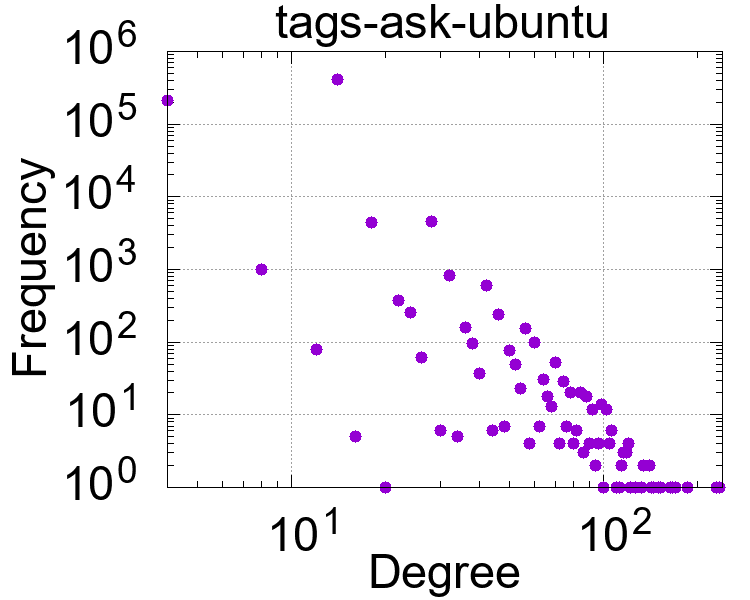} &
	\includegraphics[width=1.4in]{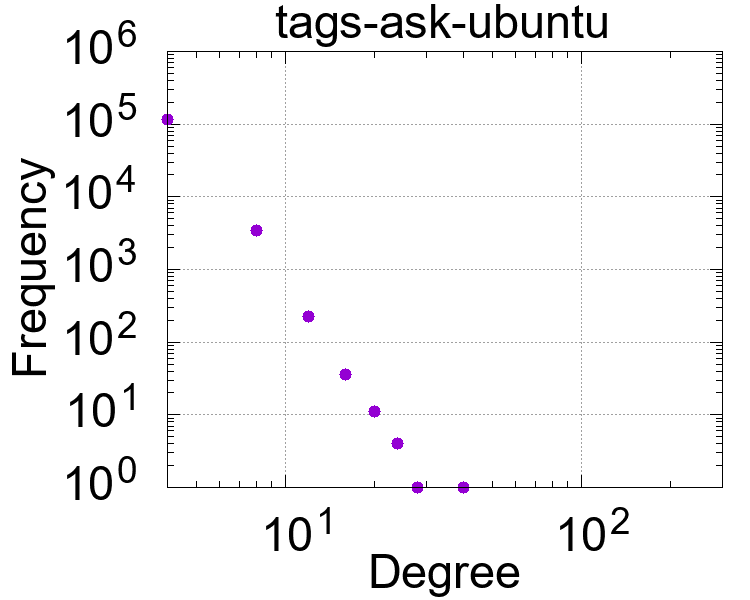}\\	
			\bottomrule
		\end{tabular}
		}
		\caption{\label{fig:summary_tags-ask-ubuntu_figure} Comparison of hypergraph generators with respect to degree distributions of decomposed graphs at different decomposition levels. The hypergraph generated by \generator resembles the real data most. See Sect.~\ref{sec:gen:exp} for numerical analysis.}
	\end{center}
\end{figure*}

%% file: 020background.tex
\begin{figure*}[t]
	\centering
	\vspace{-4mm}
	\includegraphics[width=\linewidth]{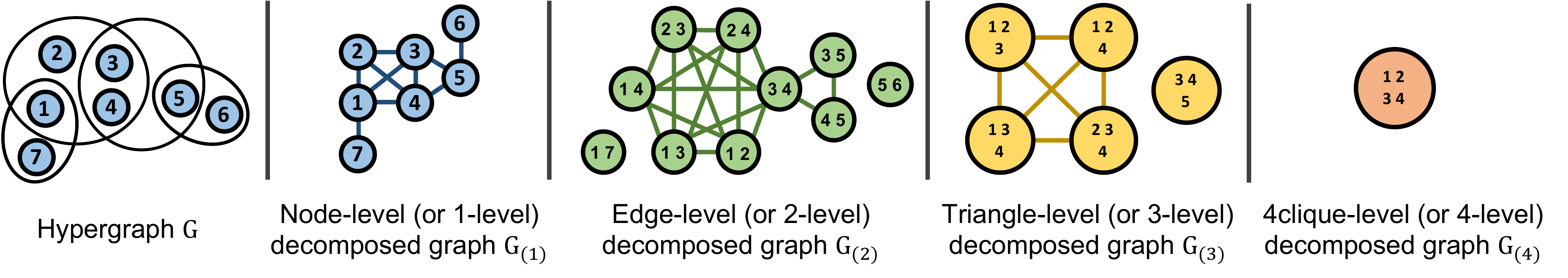} \\
	\vspace{-2mm}
	\caption{\label{fig:decomposition}
		Pictorial description of multi-level decomposition of a hypergraph $G$.
		Each $k$-level decomposed graph reveals interactions between subsets of $k$ nodes.
	}
\end{figure*}

\smallsection{Graph properties:}
Many empirical studies have been conducted to explore common properties of real-world pair-wise graphs based on predefined measurements \cite{easley2010networks}. There are two main types of these properties: static and dynamic. Static properties are revealed from a snapshot of the graphs at a particular time, and they include degree distribution \cite{abello1998functional, faloutsos1999power}, diameter \cite{bollobas2004diameter, abello2013handbook}, distribution of eigenvalues \cite{eikmeier2017revisiting}, and more \cite{yin2017local, shin2018patterns, klein1999the, milgram1967small, chung2002average, broder2000graph, bollobas2004diameter, albert1999internet, kleinberg2002small, akoglu2010oddball}. Dynamic properties examine the evolution of a graph 
over a period of time. Real-world graphs are found to possess an increasing average degree and a shrinking diameter \cite{leskovec2005graphs}. Other dynamic properties include short distances of spanning new edges \cite{leskovec2008microscopic}, temporal locality in triangle formation \cite{shin2017wrs}, and temporal network motifs \cite{paranjape2017motifs,liu2018sampling}.

\smallsection{Graph generative models:}
In conjunction, numerous graph generator models have been developed to produce synthetic graphs satisfying these commonly held patterns. Some of them focus on reproducing realistic degree distributions \cite{cooper2003general, barabasi1999emergence, mitzenmacher2004brief, mahadevan2006systematic}. Others exploit locality to generate communities within the graph \cite{kumar2000stochastic, klein1999the, kolda2014scalable, vazquez2003growing, watts1998collective, newman2001clustering}. In \cite{leskovec2005graphs,akoglu2008rtm}, dynamic patterns of graph evolution are recaptured. While most of these stochastic generator models rely on empirical results to demonstrate their abilities to repeat realistic behavior, \cite{leskovec2010kronecker,akoglu2008rtm} provide theoretical guarantees. 
Although most of the aforementioned graph generators are self-contained stochastic models, several models require some explicit fitting to real data in order to exactly reproduce the patterns \cite{sala2010measurement, edunov2016darwini,leskovec2010kronecker}.

\smallsection{Hypergraphs:}
Hypergraphs are used for representing various entities in diverse fields, including biology, medicine, social networks, and web \cite{benson2018simplicial, benson2016higher, akoglu2008rtm}. To better analyze and process hypergraphs, there has been an increasing interest in extending studies on graphs to hypergraphs, including spectral theory \cite{zhou2006learning, klein1999the} and triadic closure theory \cite{benson2018simplicial}. Studies have also proposed models of the generation and evolution of hyperedges \cite{benson2018sequences, benson2018simplicial, stasi2014beta, chodrow2019configuration}. However, \cite{benson2018sequences} focuses on repeat patterns of hyperedges, particularly on the recency bias and intensity of repeats, and generates only the next hyperedge, given all previous hyperedges.  \cite{benson2018simplicial} focuses on a particular type of hypergraph dynamics, namely simplicial closure. On the other hand, \cite{stasi2014beta} and \cite{chodrow2019configuration} try to configure the generated hypergraphs to satisfy a given degree distribution without explicitly accounting for \interactions in exploring the patterns.

In our work, we study the general patterns of real-world hypergraphs, encompassing the wide range of extensions studied in graphs with a strong emphasis on `\interactions'. On such basis, we propose and evaluate generative models for hypergraphs.

%% file: 030projection.tex
In this section, we introduce the multi-level decomposition, which is our method for analyzing hypergraphs. Our motivation for the multi-level decomposition is that it is not straightforward to investigate the properties of hypergraphs in their raw form. We instead seek a way to analyze hypergraphs through the lens of ordinary graphs. By transforming hypergraphs into graphs, we can adopt the various properties studied in graphs for hypergraphs.
%

\smallsection{Hypergraphs and \interactions:}
A \textit{hypergraph} is defined as $G=(V,E)$, where $V$ is a set of nodes and $E\subset 2^{V}$ is a set of \textit{hyperedges}. Each hyperedge $e\subseteq V$ is a set of $|e|$ nodes that have appeared as a group.
Distinguished from hyperedges, a {\it \interaction} among two or more nodes indicates their co-appearance as a subset of a hyperedge.
For example, a hyperedge $\{a,b,c,d\}$ leads to the following \interactions: $\{a,b,c,d\}$, $\{a,b,c\}$, $\{b,c,d\}$, $\{c,d,a\}$, $\{d,a,b\}$,  $\{a,b\}$, $\{a,c\}$, $\{a,d\}$, $\{b,c\}$,  $\{b,d\}$, and $\{c,d\}$.

\smallsection{Multi-level decomposition:}
Given a hypergraph $G=(V,E)$, the \textit{multi-level decomposition} of $G$ is defined as a set of $k$-level decomposed graphs for every $k \in \{1,...,\max_{e \in E}(|e|)\}$, where $\max_{e \in E}(|e|)\}$ is the maximum size of a hyperedge in $E$. 
The $k$-level decomposed graph, which is illustrated in Fig.~\ref{fig:decomposition}, is defined below.

\smallskip
\noindent\textbf{Definition 1} \textsc{($k$-level decomposed graph).} 
The $\mathit{k}$\textit{-level decomposed graph} of a hypergraph $G=(V, E)$ is $G_{(k)}=(V_{(k)}, E_{(k)})$ where
\begin{align*}
	 & V_{(k)}  := \{v_{(k)} \in 2^{V} : |v_{(k)}| = k \text{ and } \exists e\in E \text{ s.t. } v_{(k)} \subseteq e \}, \\
	 & {E_{(k)}  := \{\{u_{(k)},v_{(k)}\} \in \binom{V_{(k)}}{2} : \exists e\in E \text{ s.t. } u_{(k)} \cup v_{(k)} \subseteq e \}.}
\end{align*}

The nodes in the $k$-level decomposed graph $G_{(k)}$ of a hypergraph $G$ are the sets of $k$ nodes in $G$ that appear together in at least one hyperedge in $G$.
In $G_{(k)}$, two sets of $k$ nodes are connected by an edge if and only if there exists a hyperedge in $G$ that contains both.
That is, the $k$-level decomposed graph naturally represents how each set of $k$ nodes interacts, as a group, with other sets of $k$ nodes.\footnote{Compared to projected graphs \cite{yoon2020much}, which reveal only interactions between node sets with overlaps, decomposed graphs reveal all interactions between node sets.}
Utilizing decomposed graphs constitutes several advantages:
\begin{itemize}[$\bullet$]
\item{\bf{\Interaction}}: decomposed graphs reveal \interactions between subsets of nodes.
\item{\bf{Pairwise graph representation}}: decomposed graphs can be easily analyzed with existing measurements for pairwise graphs.
\item{\bf{No information loss}}: the original hypergraph can be recovered from the decomposed graphs (see Appendix~\ref{appendix:proof:recovery}).
\end{itemize}
Notice that the notion of $k$-level decomposition is a generalization of an existing concept: when $k=1$, the decomposed graph corresponds to the widely-used pairwise projected graph.


In our study, we focus on $k$-level decomposed graphs with $k \in \{1, 2, 3, 4\}$,
as most hyperedges in real-world hypergraphs are of sizes only up to $4$. For simplicity, we call them \textit{node-level}, \textit{edge-level}, \textit{triangle-level}, and \textit{4clique-level decomposed graphs}, respectively.







%% file: 040observation.tex
In this section, we demonstrate that the following structural patterns hold in each level of decomposed graphs of real hypergraphs\footnote{By our definition, a hyperedge of size $n > k$ results in ${n}\choose{k}$ nodes and ${{n}\choose{k}}\choose{2}$ edges in the $k$-level decomposed graph. For example, a hyperedge of 8 nodes is decomposed into ${8}\choose{3}$ $= 56$ nodes and ${56}\choose{2}$ $= 1,540$ edges in the triangle-level decomposed graph. 
In order to avoid dominance by the edges resulted from  large-size hyperedges, in the node-level decomposed graphs, only hyperedges with up to 25 nodes are considered.
In higher-level decomposition, we only consider hyperedges with up to $7$ nodes. Actually, in each dataset, the vast majority of hyperedges consist of $7$ or fewer nodes.}\textsuperscript{,}\footnote{We used Snap.py (\url{http://snap.stanford.edu/snappy}) for computing graph measures.}:
\begin{itemize}[$\bullet$]
	\item {\bf P1.} \connectcomp 
	\item {\bf P2.}  \degdistr 
	\item {\bf P3.}  \effdiam 
	\item {\bf P4.}  \clustcoef 
	\item {\bf P5.} \sngvals 
\end{itemize}

These patterns, which are described in detail in the following subsections, are supported by our observations in \textbf{thirteen} real hypergraph datasets of medium to large sizes.
Details on the datasets can be found in Appendix~\ref{appendix:description}, and the complete set of observations is available in \cite{appendix}.
Below, we provide the intuition behind them and present a random hypergraph model that we use as the null model.

\smallsection{Intuition behind the patterns.}
Consider the coauthorship data as an example: in our node-level decomposed graph, each node represents an author, and two nodes are connected if and only if these two authors have coauthored at least one paper before. 
Therefore, this node-level decomposition can be interpreted as an author network. 
Such node-level decomposed graphs are not  ``real'' graphs since they are obtained by decomposing the original hypergraphs.
However, they represent pairwise relationships as real-world graphs do, and by this interpretation, we deduce that the node-level decomposed graphs of real-world hypergraphs will exhibit the five patterns (i.e., P1-P5), which are well-known for real-world graph\blue{s} \cite{leskovec2005graphs, leskovec2010kronecker, faloutsos1999power,broder2000graph,milgram1967small, albert1999internet, bollobas2004diameter,yin2017local, kumar2000stochastic}. We further suspect that these patterns also hold at higher levels of decomposition.

\smallsection{Null Model: Random Hypergraphs (Null.):}
In order to show \textbf{P3} and \textbf{P4} are not random behavior of any hypergraph, we use a random hypergraph corresponding to each real hypergraph as the null model. Specifically, given a hypergraph, the null model is constructed by randomly choosing nodes to be contained in each hyperedge, while keeping its original size.

%

\vspace{-1mm}
\subsection{P1. Giant connected component}
\label{sec:pattern:giant}
\vspace{-1mm}

This property means that there is a connected component comprising of a large proportion of nodes, and this proportion is significantly larger (specifically, at least $70$ times larger) than that of the second largest connected component. The majority of nodes in a network are connected to each other \cite{kang2010patterns}. This property serves as a basis for the other properties. For example, without a giant connected component (i.e, the graph is ``shattered'' into small connected communities), 
diameter would clearly be small as a consequence, not as an independent property of the dataset.
\input{046cshatterlevel.tex}

In Table~\ref{tab:gcc}, we report the size of the largest connected component at all decomposition levels. The connectivity of \interactions, represented as the highest level for which the decomposed graph maintains a giant connected component, varies among datasets. In particular, while the co-authorship datasets are shattered at the triangle level, the online-tags datasets retain giant connected components until the 4clique level. Note that while our decomposition is only up to the 4clique level, there are many hyperedges of sizes at least 5, implying that when the graph is shattered, it consists of several isolated cliques, not just isolated nodes.

There is a positive correlation between the distribution of hyperedge sizes and whether the graph is shattered at the edge-level decomposition. Take the proportion of unique hyperedges of sizes at most $2$ as the feature. Datasets with this feature greater than $75\%$ are shattered, and the others retain giant connected components. At the triangle level, $6$ (out of $13$) datasets have giant connected components.
Except for \textit{email-Eu} and \textit{NDC-classes}, the datasets where the proportion of hypergedges of sizes at most 3 is larger than $60\%$ are shattered at this level. The others possess a giant connected component. 

\input{047summary_degrees_sng_figure.tex}
\input{041cnodeproperties.tex}

\vspace{-1mm}
\subsection{P2. Heavy-tailed degree distribution}
\label{sec:pattern:deg}
\vspace{-1mm}

The degree of a node is defined as the number of its neighbors. This property means that the degree distribution is heavy-tailed, i.e decaying at a slower rate than the exponential distribution (exp.). This can be partially explained by the ``rich gets richer":  high-degree nodes are more likely to form new links \cite{newman2001clustering}. Besides visual inspection, we confirm this property by the following two tests:
\begin{itemize}[$\bullet$]
\item Lilliefors test \cite{lilliefors1969kolmogorov} is applied at significance level $2.5\%$ with the null hypothesis $H_{0}$ that the given distribution follows exp.
\item The likelihood method in \cite{clauset2009power, alstott2014powerlaw} is used on the given distribution to compute the likelihood ratio $r$ of a heavy-tailed distribution (power-law, truncated power-law or lognormal) against exp. If $r>0$, the given distribution is more similar to a heavy-tailed distribution than exp.
\end{itemize}

In Fig.~\ref{fig:summary_degree_sng_figure},
we illustrate that for each dataset, at the decomposition level in which there is a giant connected component, the degree distribution is heavy-tailed.
Applying the two tests, in all cases, either $H_{0}$ is rejected or $r>0$ (both claims hold in most cases), indicating evidence for heavy-tailed degree distribution.\footnote{In \textit{coauth-DBLP}, at the edge level, $H_{0}$ is accepted at $2.5\%$ significance level, but the loglikelihood ratios of the heavy-tailed distributions over exp. are greater than 5000.} The loglikelihood ratios are reported in Table~\ref{tab:heavy_tailed_deg_main}. Except for \textit{email-Eu} at the node level, in all cases, at least one heavy-tailed distribution has a positive ratio, implying that the degree distribution is more similar to that distribution than it is to exp.

\input{090summary_loglikelihood_ratio_deg_main.tex}

\vspace{-1mm}
\subsection{P3. Small diameter}
\label{sec:pattern:diam_clust} 
\vspace{-1mm}

Decomposed graphs are usually not completely connected, and it makes diameter subtle to define. 
We adopt the definition in \cite{leskovec2005graphs}, where the effective diameter is the minimum distance $d$ such that approximately $90 \%$ of all connected pairs are reachable by a path of length at most $d$. 
This property means that the effective diameter in real datasets is relatively small, and most connected pairs can be reachable by a small distance \cite{watts1998collective}.
Note that the null model also possesses this characteristic, 
and comparing real-world datasets and the corresponding null model in this aspect does not yield consistent results. 
The effective diameters at the 4 decomposition levels are highlighted in Tables~\ref{tab:node_properties} and \ref{tab:edge_triangle_4clique_properties}.

\vspace{-1mm}
\subsection{P4. High clustering coefficient}
\label{sec:pattern:clust}
\vspace{-1mm}

We make use of the clustering coefficient $C$ \cite{watts1998collective}, defined as the average of local clustering coefficients of all nodes. The local clustering coefficient $C_v$ of each node $v$ is defined as:
 $$C_{v} := 2\times \frac{\text{the number of triangles involving $v$}}{\text{the number of connected triples of nodes involving $v$}}.$$
This property means that the statistic in the real datasets is significantly larger than that in the corresponding null models. As communities result in a large number of triangles, this property implies the existence of many communities in the network. 

In Table~~\ref{tab:node_properties}, clustering coefficients of the datasets are compared against that of the corresponding null model at the node-level decomposition.
From the edge level, the decomposed graph of the null model is almost shattered into small isolated cliques. As a result, the clustering coefficient is unrealistically high, making it no longer valid to compare this statistic to that of the real-world data. Results at the edge or higher-level decompositions are reported in Table~\ref{tab:edge_triangle_4clique_properties}.

\vspace{-1mm}
\subsection{P5. Skewed singular values}
\label{sec:pattern:sng}
\vspace{-1mm}

This property means the singular-value distribution is usually heavy-tailed, and it is verified in the same manner as the pattern \textbf{P2}.
In all cases where a giant connected component is retained, 
either $H_{0}$ is rejected or the log likelihood ratio $r > 0$, implying that the singular-value distributions are heavy-tailed.
Specifically, as seen in Table~\ref{tab:heavy_tailed_deg_main}, except for \textit{tags-stack-overflow} at the edge level, in all cases, at least one heavy-tailed distribution has a positive ratio.
Some representative plots for singular-value distributions of real datasets are provided in Fig.~\ref{fig:summary_degree_sng_figure}. 






\input{048summary_properties.tex}

To support the patterns {\bf P1}-{\bf P5}, we could provide only some representative results in Tables~\ref{tab:node_properties}-\ref{tab:edge_triangle_4clique_properties}, and Fig.~\ref{fig:summary_degree_sng_figure} due to the space limit. The complete set of figures and numerical data 
 can be found in \cite{appendix}.  

%% file: 046cshatterlevel.tex
\begin{table}[t]
\vspace{-3mm}
\centering
\caption{\label{tab:gcc} Size of the largest connected component, as the proportion of the total number of nodes (including the degree-zero nodes), in each dataset at each decomposition level. The red numbers indicate that the graph no longer retains a giant connected component. In the case of \textit{NDC-classes}, the size of the second largest connected component at triangle and 4clique levels is 0.11 and 0.04. According to the description in Sect.~\ref{sec:pattern:giant}, a giant connected component does not exist.} 
\scalebox{0.95}{
\begin{tabular}{l|cccc}
\toprule
Level & Node      & Edge     & Triangle  & 4clique\\ 
& ($k=1$) & ($k=2$) & ($k=3$) & ($k=4$) \\
\midrule
coauth-DBLP            &  0.86    & 0.57   & \red{0.05}   & \red{0.0006}      \\
coauth-Geology         &  0.72    & 0.5   & \red{0.06}  & \red{0.0005}	  \\
coauth-History         &  0.22    & \red{0.002} & \red{0.002}  & \red{0.001}		  \\
DAWN                   &  0.89    & 0.98   &  0.91   &      0.52 	  \\
email-Eu               &  0.98    & 0.98   &  0.86   & 		0.41	  \\
NDC-classes            &  0.54    & 0.62   &  \red{0.27}   & \red{0.19}	  \\
NDC-substances         &  0.58    & 0.82   &  0.36   & 	\red{0.02}	  \\
tags-ask-ubuntu        &  0.99    & 0.99   &  0.79   & 		0.21	  \\
tags-math              &  0.99    & 0.99   &  0.91   & 		0.35	  \\
tags-stack-overflow    &  0.99    & 0.99   &  0.92   & 		0.42	  \\
threads-ask-ubuntu     &  0.65    & \red{0.09}   &  \red{0.02}   & \red{0.01}	  \\
threads-math           &  0.86    & 0.61   &  \red{0.03}   & \red{0.0004}	  \\
threads-stack-overflow &  0.86    & 0.32   & \red{0.004}  & \red{$3\mathrm{e}{-5}$}		  \\
\bottomrule
\end{tabular}
}
\end{table}

%% file: 047summary_degrees_sng_figure.tex
\begin{figure*}[t]
	\addtolength{\tabcolsep}{-3pt}
	\begin{center}
	\vspace{-4mm}
	\centering
	\scalebox{0.75}{
  \begin{tabular}{c|ccc|ccc}
  	\toprule
  	{Levels}  & \multicolumn{3}{c|}{Degree Distributions of Decomposed Graphs} 
  	& \multicolumn{3}{c}{Singular-value Distributions of Decomposed Graphs}
  	\\ 
	\midrule  
    \rotatebox{90}{\hspace{1.1cm}Node}& 
    \includegraphics[width=1.45in]{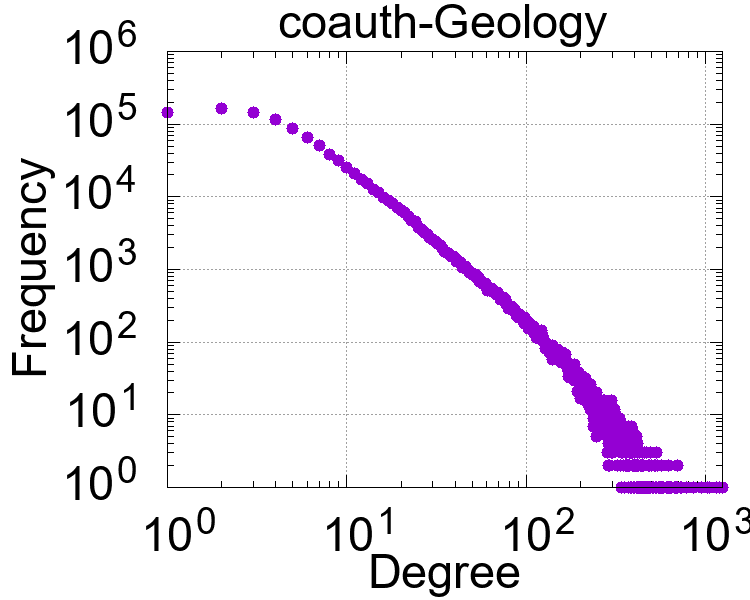} &
    \includegraphics[width=1.45in]{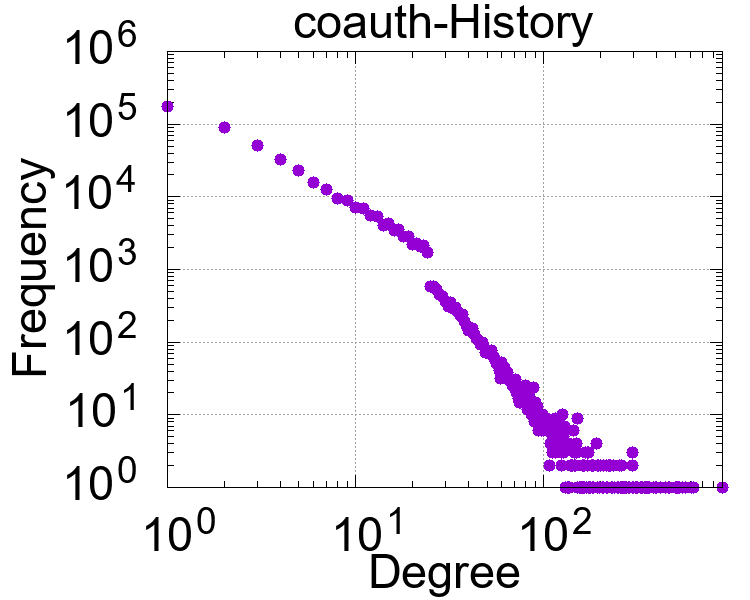} &
    \includegraphics[width=1.45in]{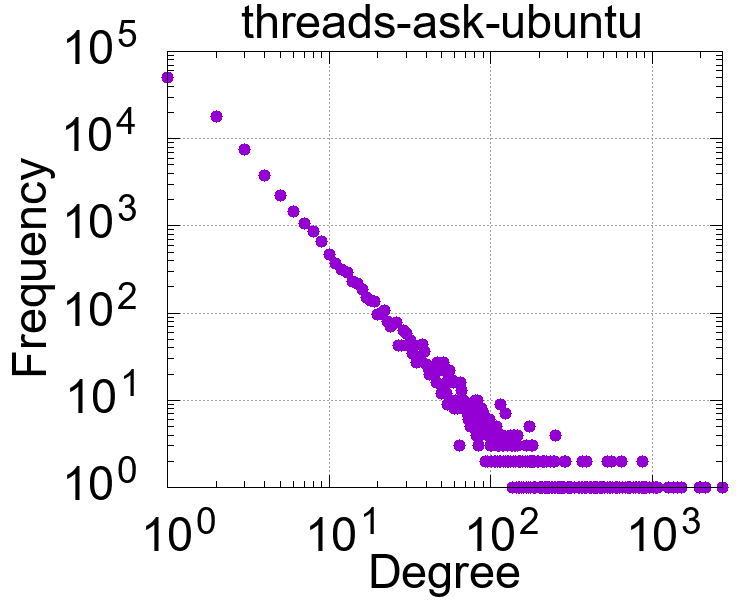} &
    \includegraphics[width=1.45in]{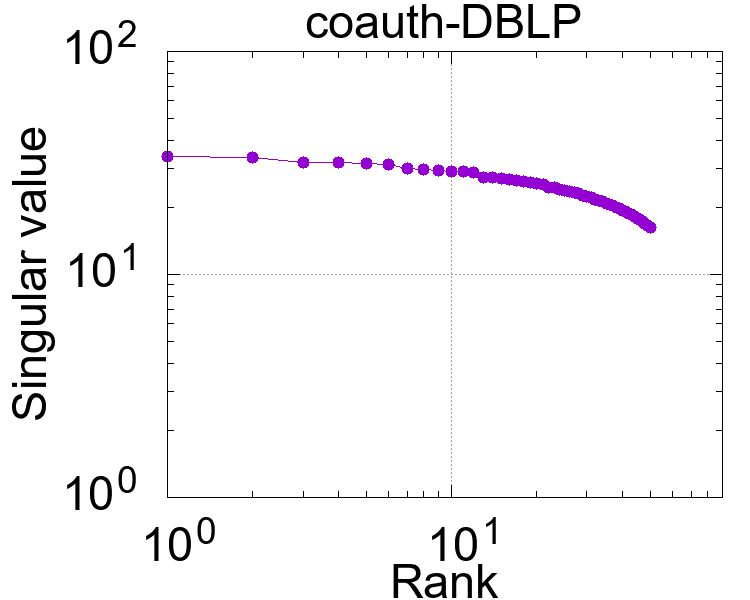} &
    \includegraphics[width=1.45in]{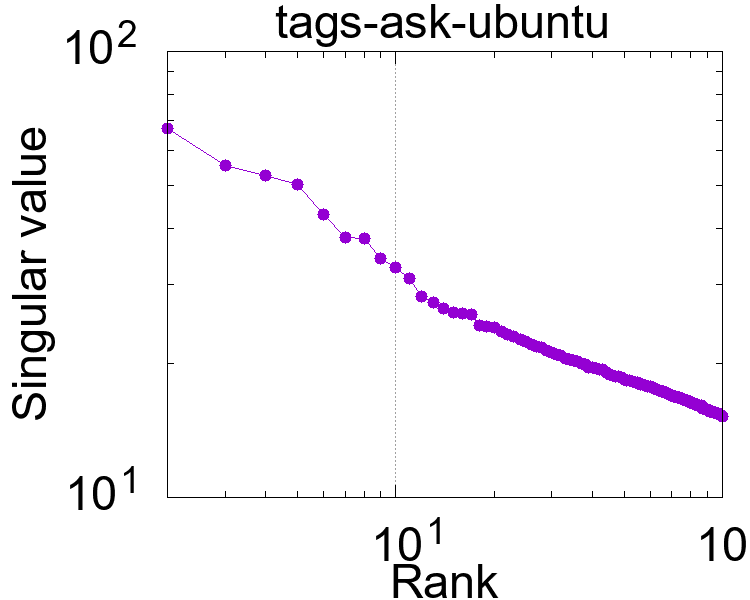} &
    \includegraphics[width=1.45in]{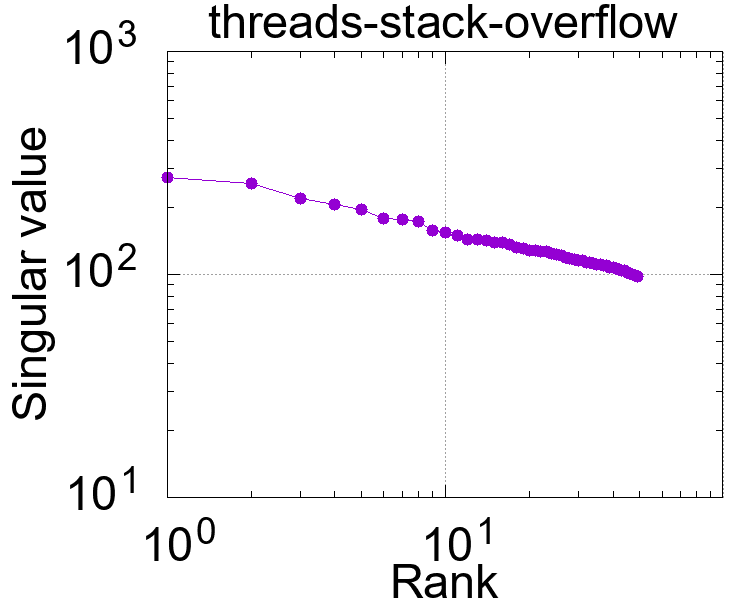} \\
    \midrule
    \rotatebox{90}{\hspace{1.1cm}Edge}&
    \includegraphics[width=1.45in]{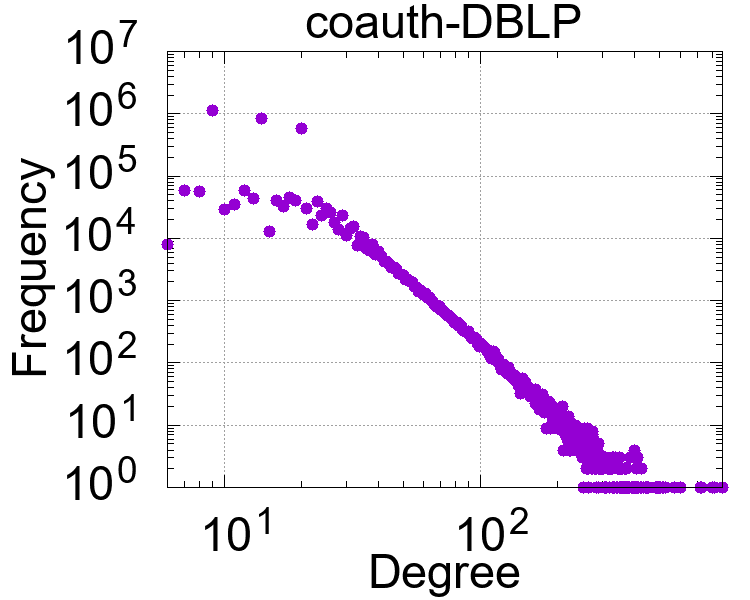} &
    \includegraphics[width=1.45in]{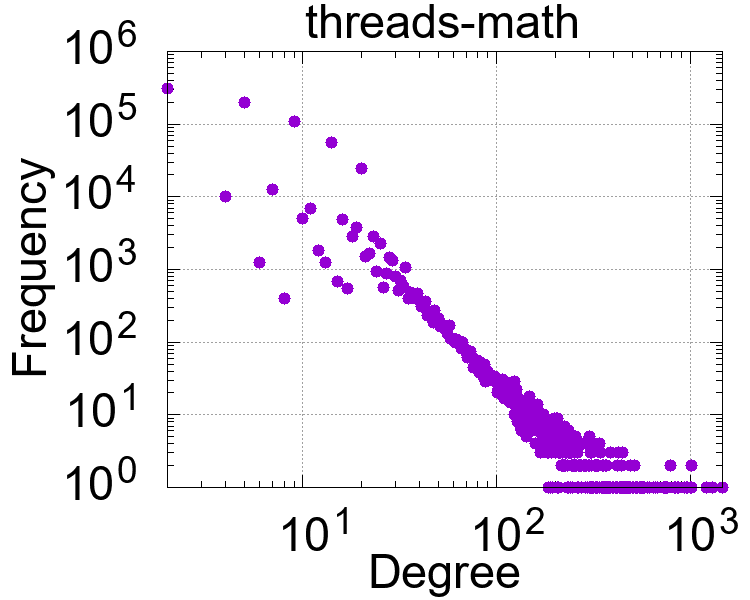} &
    \includegraphics[width=1.45in]{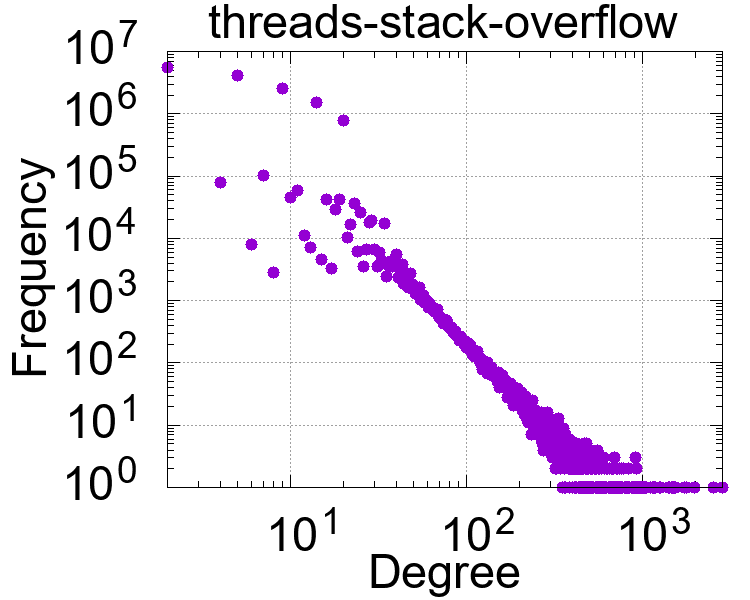} &
    \includegraphics[width=1.45in]{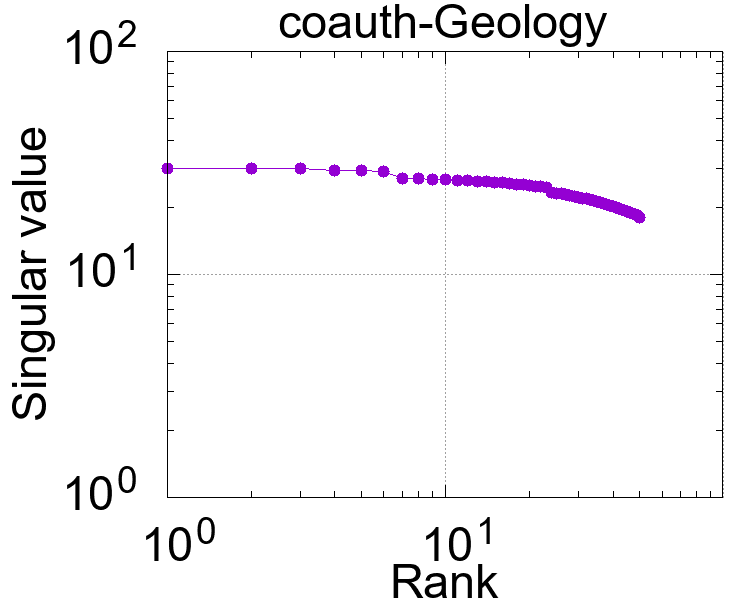} &
    \includegraphics[width=1.45in]{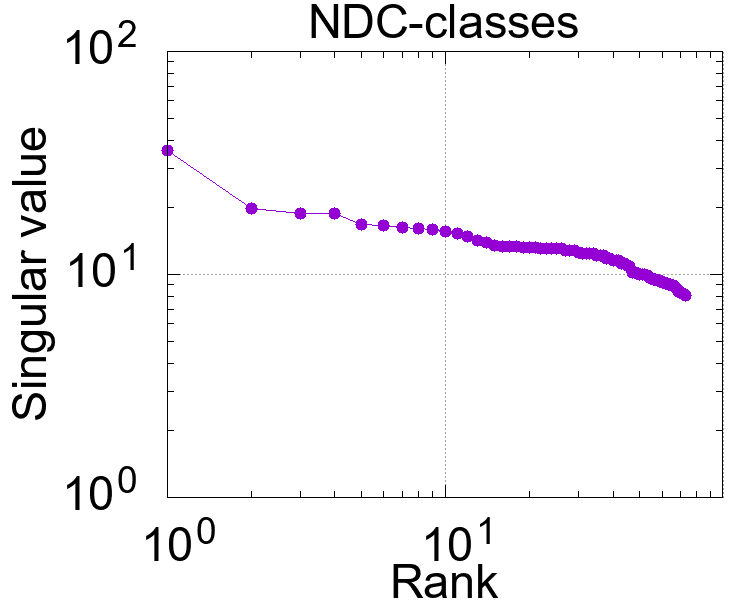} &
    \includegraphics[width=1.45in]{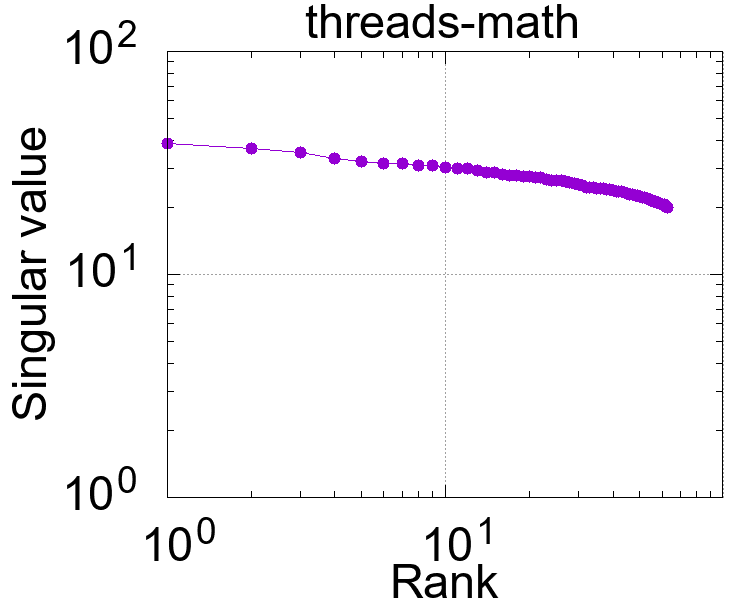} \\
    \midrule
    \rotatebox{90}{\hspace{1cm}Triangle} & 
    \includegraphics[width=1.45in]{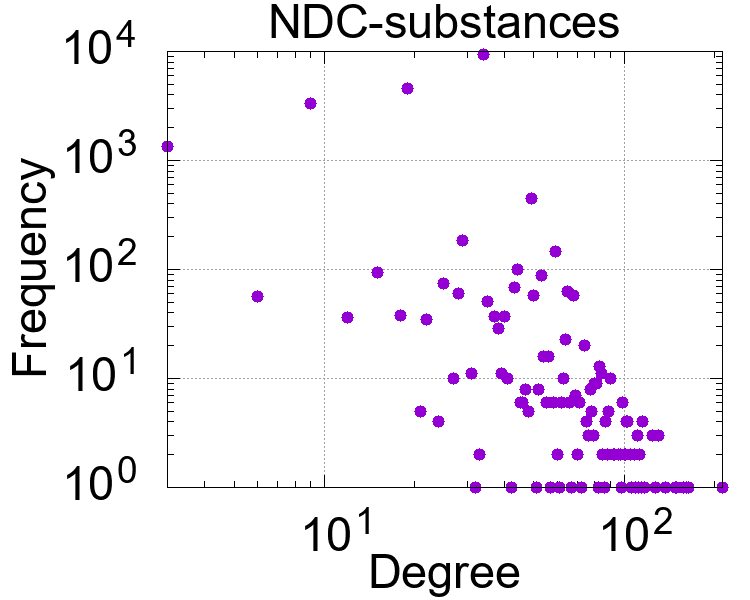} &
    \includegraphics[width=1.45in]{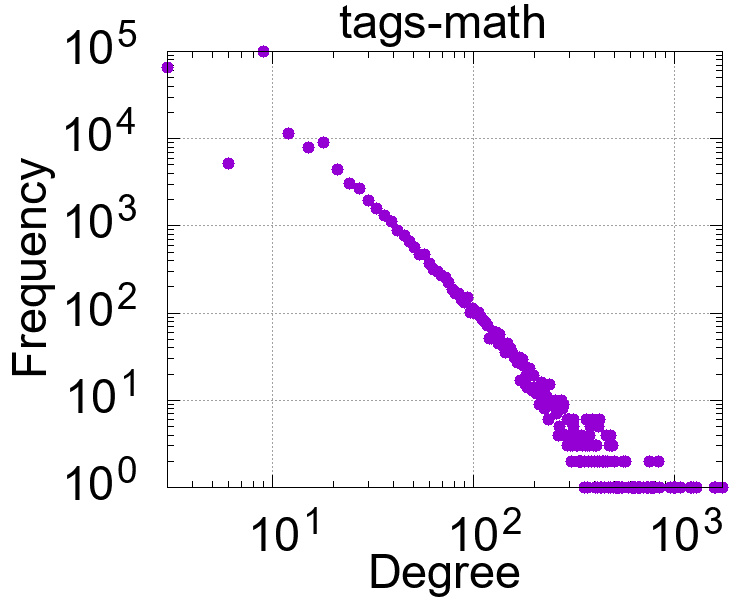} &
    \includegraphics[width=1.45in]{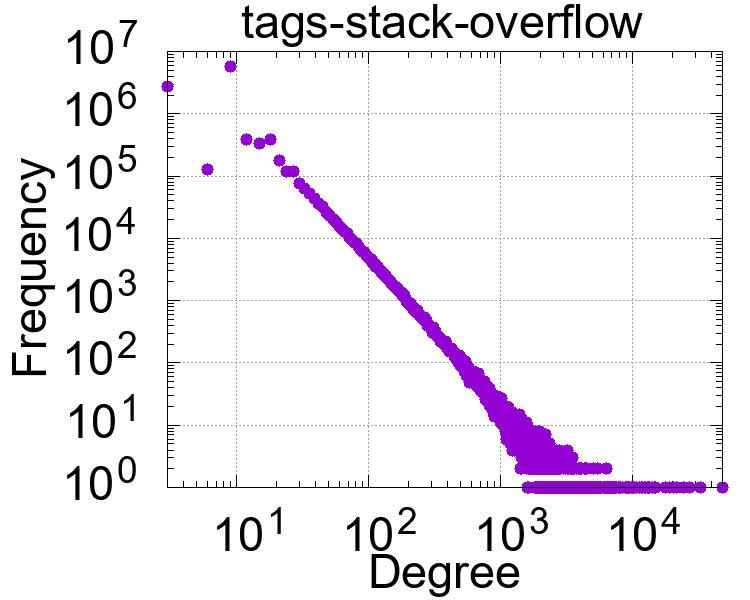} &
    \includegraphics[width=1.45in]{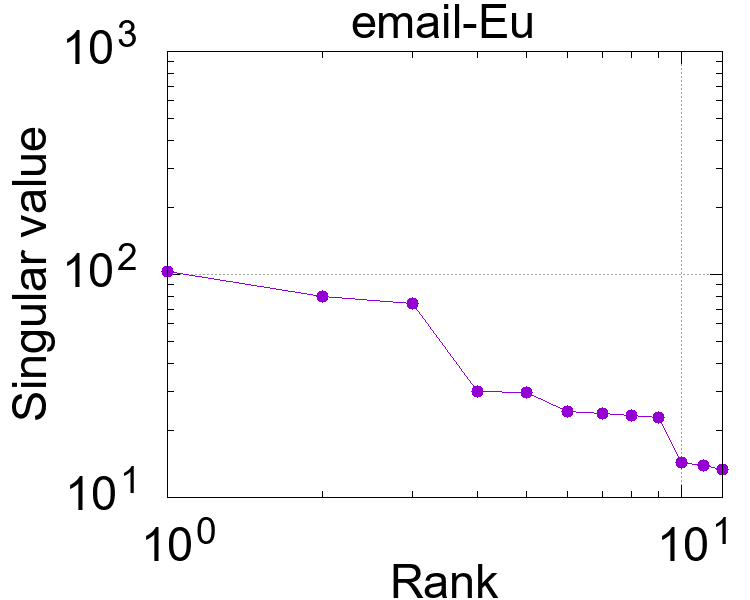} &
    \includegraphics[width=1.45in]{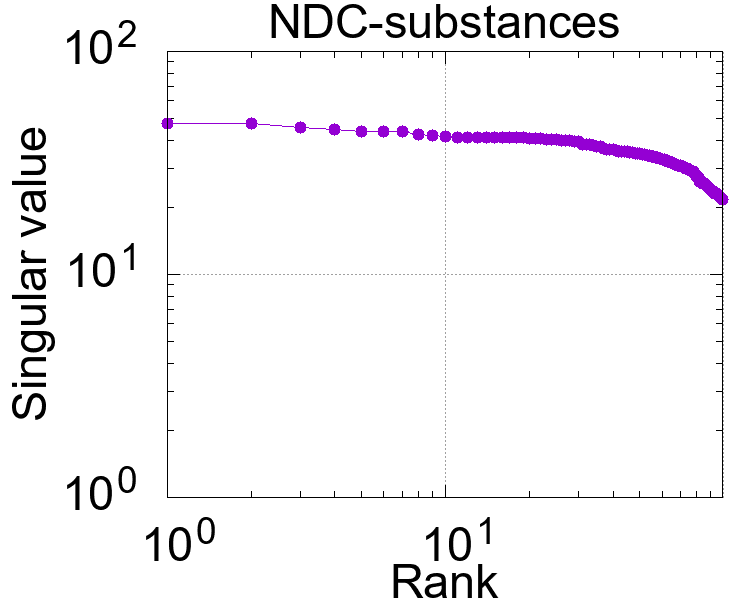} &
    \includegraphics[width=1.45in]{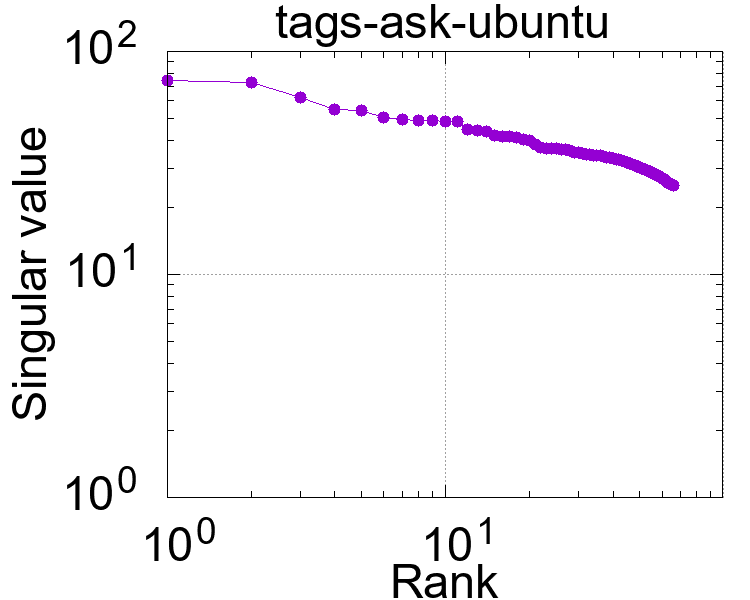} \\
    \midrule
    \rotatebox{90}{\hspace{1cm}4clique} & 
    \includegraphics[width=1.45in]{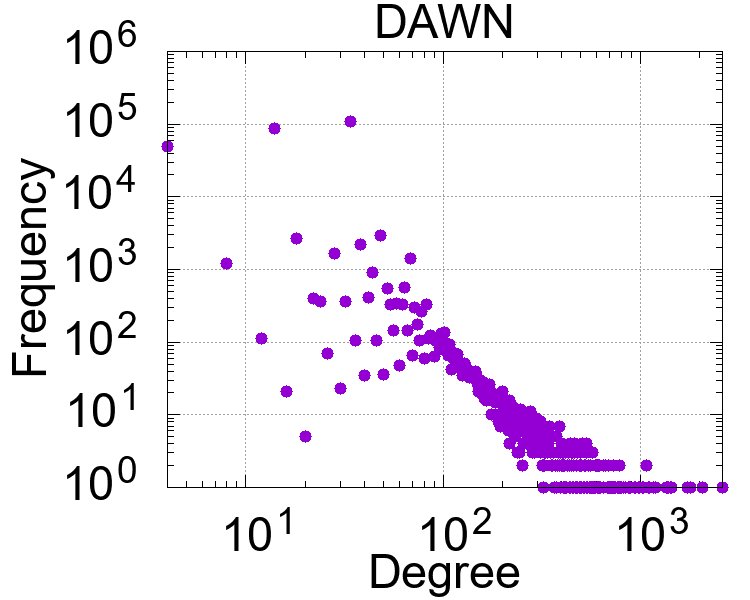} &
    \includegraphics[width=1.45in]{FIG/4cliqueDeg_tags-ask-ubuntu_.png} &
    \includegraphics[width=1.45in]{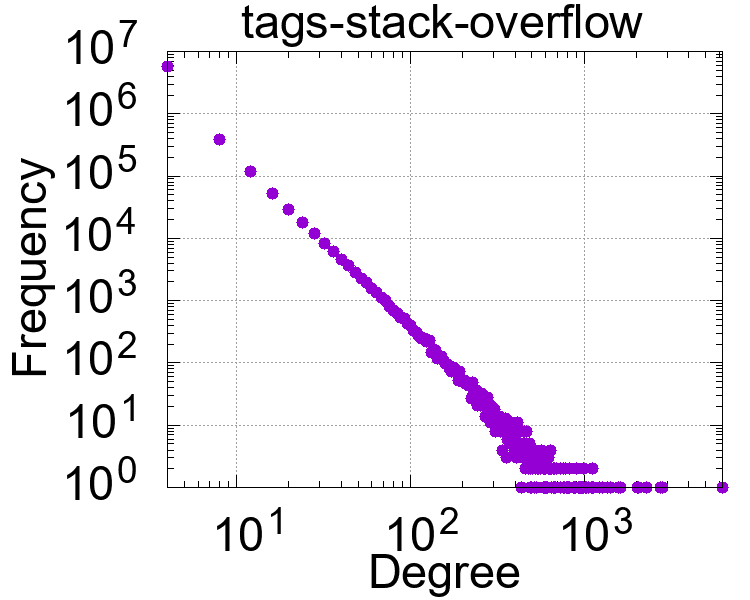} &
    \includegraphics[width=1.45in]{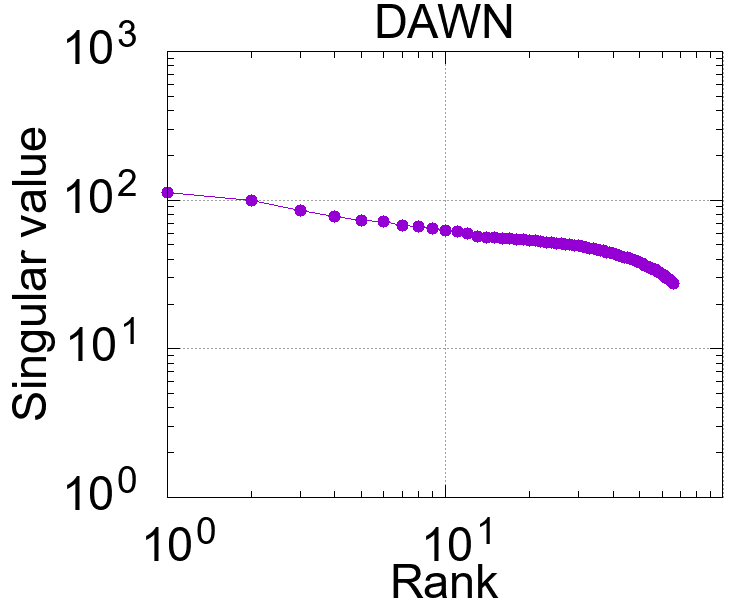} &
    \includegraphics[width=1.45in]{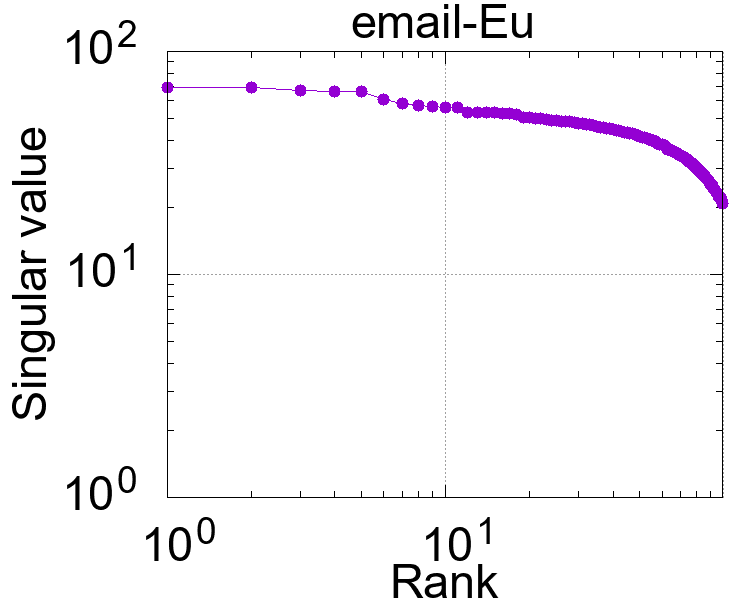} &
    \includegraphics[width=1.45in]{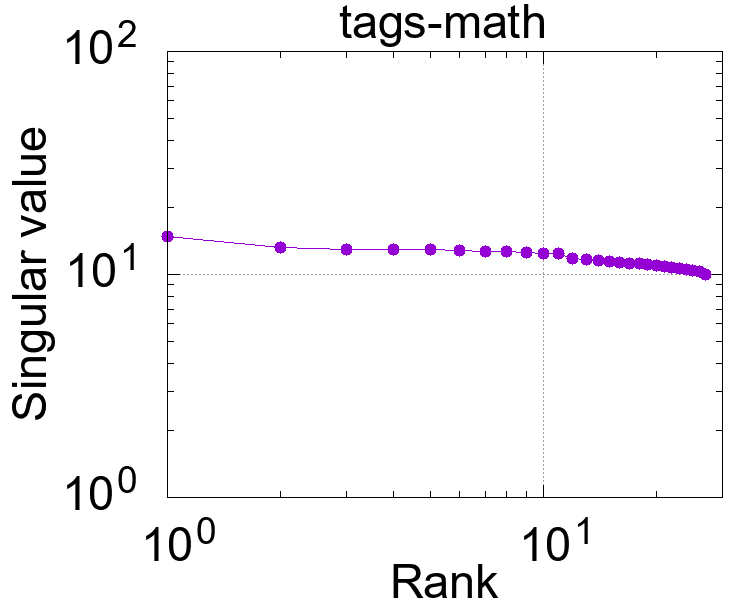} \\
    \bottomrule
  \end{tabular}
  }
  \caption{\label{fig:summary_degree_sng_figure}  	
  Representative plots for the degree and singular-value distributions of decomposed graphs at 4 decomposition levels. They are heavy-tailed, and specifically on the log-log scale, tails often approximate a straight line.}
  
\end{center}
\end{figure*}

%% file: 041cnodeproperties.tex
\begin{table}[t]
\vspace{-3mm}
\centering
\caption{\label{tab:node_properties} Properties of \nodeproj of all datasets. The \textit{diameter} and \textit{clustering coefficient} are compared against a null model. Average and standard deviation of 10 random hypergraphs are reported. All node-level decomposed graphs possess a diameter relatively small to the number of nodes. Almost all of them have clustering coefficients significantly higher than that of the null model. 
}
  \scalebox{0.83}{
  \begin{tabular}{l|ccccc}
    \toprule
    {\bf Dataset}
    & {\bf \# Nodes}
    & \multicolumn{2}{c}{\bf Eff. diameter}
    & \multicolumn{2}{c}{\bf Clust. coeff.}
     \\
	
    \cmidrule(lr){3-4}
    \cmidrule(lr){5-6}
    
    & 
    & Real & Null.
    & Real & Null.\\
    \cmidrule(lr){1-1}
    \cmidrule(lr){2-2}
    \cmidrule(lr){3-4}
    \cmidrule(lr){5-6}
     coauth-DBLP
    & 1,924,991  & 6.8 & 6.7 $\pm9\mathrm{e}{-3}$ & 0.60 & 0.31 $\pm1\mathrm{e}{-4}$\\
    coauth-Geology
    & 1,256,385  & 7.1 & 6.8 $\pm8\mathrm{e}{-3}$ & 0.57 & 0.42  $\pm2\mathrm{e}{-4}$\\
    coauth-History
    & 1,014,734  & 11.9 & 17  $\pm0.19$& 0.24 & 0.26  $\pm2\mathrm{e}{-4}$\\
    DAWN
    & 2,558  & 2.6 & 1.85 $\pm8\mathrm{e}{-5}$ & 0.64 & 0.30  $\pm9\mathrm{e}{-5}$\\
    email-Eu
    & 998  & 2.8 & 1.85 $\pm7\mathrm{e}{-5}$ & 0.49 & 0.36  $\pm5\mathrm{e}{-4}$\\
    NDC-classes 
    & 1,161  & 4.6 & 2.6 $\pm6\mathrm{e}{-3}$ & 0.61 & 0.32  $\pm2\mathrm{e}{-3}$\\
    NDC-substances 
    & 5,311  & 3.5 & 2.5 $\pm9\mathrm{e}{-3}$ & 0.40 & 0.17  $\pm6\mathrm{e}{-4}$\\
    tags-ask-ubuntu
    & 3,029  & 2.4 & 1.9 $\pm2\mathrm{e}{-5}$ & 0.61 & 0.14  $\pm7\mathrm{e}{-5}$\\
    tags-math
    & 1,629  & 2.1 & 1.8 $\pm1\mathrm{e}{-4}$ & 0.63 & 0.46  $\pm2\mathrm{e}{-4}$\\
    tags-stack-overflow
    & 49,998  & 2.7 & 1.9 $\pm2\mathrm{e}{-6}$ & 0.63 & 0.03  $\pm1\mathrm{e}{-6}$\\
    threads-ask-ubuntu 
    & 125,602  & 4.7 & 11.9 $\pm0.042$ & 0.11 & 0.19  $\pm7\mathrm{e}{-4}$ \\
    threads-math
    & 176,445  & 3.7 & 4.9 $\pm4\mathrm{e}{-3}$ & 0.32 & 0.12  $\pm1\mathrm{e}{-4}$\\
    threads-stack-overflow
    & 2,675,995  & 4.5 & 5.9 $\pm2\mathrm{e}{-3}$ & 0.18 & 0.12  $\pm2\mathrm{e}{-5}$\\
  \bottomrule
  \end{tabular}
  }
\end{table}

%% file: 090summary_loglikelihood_ratio_deg_main.tex
\begin{table}[h!]
\caption{\label{tab:heavy_tailed_deg_main} Loglikehood ratio when fitting the degree and singular-value distributions to each of three heavy-tailed distributions versus the exponential distribution.
In most cases, there exists at least one positive ratio, implying that both distributions are heavy-tailed. Due to underflow problems, the results for truncated power-law are not available in some cases.
}
  \scalebox{0.77}{
  \begin{tabular}{l|cccccc}
    \toprule
    Measure
    & \multicolumn{3}{c}{\bf Degree}
    & \multicolumn{3}{c}{\bf Singular values}
    \\
     \cmidrule(lr){2-4}
     \cmidrule(lr){5-7}
     Heavy-tail distribution 
     & pw & trunpw & lgnorm  
     & pw & trunpw & lgnorm 
     \\
    \midrule
    \multicolumn{4}{l}{\textbf{Node-level decomposed graphs}} \\    
    \midrule
     coauth-DBLP 
     & \bf 1108 & \bf 1108& \bf 1108
     & \bf 3.4 & \bf 3.6 & \bf 6.4 \\
     coauth-Geology 
     & \bf 10.77& \bf 11.3& \bf 11.3
     & -2.3  & - & \bf 11.3\\
     coauth-History 
     & \bf 429& \bf 430& \bf 429.9 
     & -1 & -0.07& \bf 0.3  \\
     DAWN 
     & -4.9 &-0.5 & \bf 0.3 
     & \bf 16.8  &\bf 16.8 & \bf 22 \\
     email-Eu 
     & -15.3& -1.3&-1.1  
     & -1.3  & -0.14 & \bf 0.4 \\
     NDC-classes 
     & \bf 2.17&\bf 18.9 &\bf 14.3  
     & \bf 1.2 &\bf 1.3 &\bf 1.3  \\
     NDC-substances 
     & -8&\bf 24.8 &\bf 20.5  
     & \bf 7.5 &\bf 7.5 &\bf 11.8 \\
     tags-ask-ubuntu 
     & -1.4& \bf 6.1&\bf 4.9  
     & \bf 9.5 & \bf 9.5& \bf 9.5  \\
     tags-math 
     & -11.4& \bf 0.37&-1.1  
     & \bf 9.9 & \bf 9.9 & \bf 9.9  \\
    
     tags-stack-overflow 
     &\bf 202.8 &\bf 245.1 &\bf 241.8  
     &\bf 6 &\bf 6 &\bf 6.1  \\
     threads-ask-ubuntu 
     & \bf 2322& \bf 2330&\bf 2326  
     & \bf 2.3& \bf 2.3 &\bf 5.7  \\
     threads-math 
     & \bf 67574&\bf 67751 &\bf 67725 
     & \bf 6.6 &\bf 6.6 &\bf 11.4 \\
     threads-stack-overflow 
     &\bf 2486  & \bf 2549&\bf 2543  
     &\bf 2.1  & \bf 2.1 &\bf 2.1 \\
     \midrule
     \multicolumn{4}{l}{\textbf{Edge-level decomposed graphs}} \\ 
     \midrule
     coauth-DBLP
     & \bf 5616 & \bf 5735& \bf 5718 
     & \bf 1.3 & \bf 1.3 & \bf 4.5\\
     coauth-Geology
     & \bf 122.1& \bf 123.3 & \bf 123.4
     & \bf 122.1 & \bf 123.3 & \bf 123.4\\
     DAWN
     & \bf 4025& \bf 4389& \bf 4303 
     & \bf 0.5 & \bf 0.6 & \bf 0.5 \\
     email-Eu
     &\bf 10.9 &\bf 11.8 &\bf 11.5 
     & -1.3  & -0.14 & \bf 0.4  \\
     NDC-classes
      & \bf 44.9& \bf 44.9& \bf 44.9
      & \bf 1.2 &\bf 1.3 &\bf 1.3  \\
     NDC-substances
     & \bf 10.9& \bf 21.8&\bf 19.4 
     & \bf 10.9& - &\bf 0.3 \\
     tags-ask-ubuntu
     & \bf 36.1&\bf 41.3 &\bf 39.7 
     & -0.6&\bf 0.14 &\bf 0.05 \\
     tags-math
     & \bf 20.4& \bf 24& \bf 23.6
     & -1.3 & \bf 0.01 & -0.1\\
     tags-stack-overflow
     & \bf 394268&\bf 395917 &\bf 395852 
      & -1.5& - & -0.15 \\
     threads-math
     &\bf 1524 & \bf 1534&\bf 1528 
     &\bf 0.44 & \bf 0.44& \bf 3 \\
     threads-stack-overflow
      &\bf 4760 &\bf 4785 &\bf 4775
      &-2.6  &-0.3 &\bf 4.3 \\
     \midrule
     \multicolumn{4}{l}{\textbf{Triangle-level decomposed graphs}} \\ 
     \midrule
     DAWN
     &\bf 1392 &\bf 1426 &\bf 1417 
     &\bf 3.3 &\bf 3.3 &\bf 3.3 \\
     email-Eu
     &\bf 6.8 &\bf 6.9 & \bf 6.8
      &-1.2 & -0.12 & \bf 0.4\\
     NDC-substances
     & \bf 0.6& \bf 0.6 & \bf 0.6 
     & -4 &-0.5 & \bf 12.6\\
     tags-ask-ubuntu
     & \bf 378.6&\bf 383.2 &\bf 381 
     & -0.4 &\bf 0.15 &\bf 0.3 \\
     tags-math
      & \bf 96.4& \bf 100.8& \bf 99.3
       & -0.03 & \bf 0.001 & -0.001\\
     tags-stack-overflow
      & \bf 33198&\bf 33351 &\bf 33319 
     & -0.5 &\bf 0.1  &\bf 0.1 \\
     \midrule
     \multicolumn{4}{l}{\textbf{4clique-level decomposed graphs}} \\ 
     \midrule
     DAWN
     & \bf 372.6& \bf 377.8&\bf 374.4
     & \bf 0.04 & \bf 0.2 &\bf 0.2 \\
     email-Eu
      & -2& \bf 0.15&-0.19 
      & -0.8 & -0.07 &\bf  0.4 \\
     tags-ask-ubuntu
      & \bf 21.5&\bf 21.5 &\bf 25.9 
       & -0.36 &-0.04  &\bf 0.54 \\
     tags-math
      & \bf 107.5&\bf 107.5 & \bf 112
      & -0.06 & - & \bf 0.13\\
     tags-stack-overflow
     &\bf 31.6 & \bf 31.6& \bf 31.6
     &\bf 31.6 & \bf 31.6& \bf 31.6\\
  \bottomrule
  \end{tabular}
  }

\end{table}

%% file: 048summary_properties.tex
\begin{table}[ht!]
	\centering
	\caption{\label{tab:edge_triangle_4clique_properties}
		Numerical properties of edge or higher-level decomposed graphs of real-world datasets. 
		As the decomposition level increases, fewer datasets retain giant connected components, and the properties of such datasets are reported in the table. 
		In them, small diameters and high clustering coefficients are observed.
	}
	\scalebox{0.9}{
		\begin{tabular}{l|cccc}
			\toprule
			Measure                & Nodes    & Connect.   &Eff.    & Clust.    \\
			                       &          & Comp.      & Diam.  & Coeff. 
			\\ 
			\midrule
			\multicolumn{4}{l}{\textbf{\Edgeproj}} \\ 
			\midrule
			coauth-DBLP            & 5,906,196  & 0.57      &18.6  & 0.93  \\
			coauth-Geology     & 3,175,868  & 0.50      &16.4   & 0.94  \\
			DAWN                   &  72,288   & 0.98      &3.9    & 0.72  \\
			email-Eu               &   13,499  & 0.98      &5.71   & 0.81  \\
			NDC-classes            &  2,658    & 0.62      &6.6   & 0.94  \\
			NDC-substances         &   12,882  & 0.812     &9.4   & 0.89  \\
			tags-ask-ubuntu        &  126,518  & 0.98      &4.5    & 0.75  \\
			tags-math              &   88,367  & 0.99      &3.9    & 0.71  \\
			tags-stack-overflow    &  4,083,464 & 0.99      &3.9    & 0.78  \\
			threads-math           &  782,102  & 0.61      &7.4    & 0.94  \\
			threads-stack-overflow & 15,108,684 & 0.32      &12    & 0.97  \\
			\midrule
			\multicolumn{4}{l}{\textbf{\Triangleproj}} \\ 
			\midrule
			DAWN                   & 257,416    & 0.91      & 5.3	      & 0.87	  \\
			email-Eu               & 24,993     & 0.86      & 10.3         & 0.89	  \\
			NDC-substances         & 20,729     & 0.36      & 9.4   	      & 0.96	  \\
			tags-ask-ubuntu        & 248,596    & 0.79      & 7.8  	      & 0.89	  \\
			tags-math              & 222,853    & 0.91      & 6.7   	      & 0.85	  \\
			tags-stack-overflow      & 10,725,751  & 0.92      & 6.5	      & 0.88	  \\
			\midrule
			\multicolumn{4}{l}{\textbf{\Fcliqueproj}} \\
			\midrule
			DAWN                   & 284,755    & 0.52      & 8.1        & 0.89		  \\
			email-Eu               & 24,772     & 0.41      & 15.3       & 0.89		  \\
			tags-ask-ubuntu        & 145,676    & 0.22      & 17.1       & 0.74		  \\
			tags-math              & 156,129    & 0.35      & 14.8       & 0.71		  \\
			tags-stack-overflow    & 7,887,748   & 0.42      & 13         & 0.76		  \\
			\bottomrule
		\end{tabular}
	}
\end{table}

%% file: 050generator.tex
We have shown that five common properties of real-world pairwise graphs are revealed at different levels of decomposition of real-world hypergraphs. In this section, we present \generator, our proposed hypergraph generator model. By analyzing several statistics, we demonstrate that \generator can exhibit the known properties at several levels of decomposition. Compared to two baseline models, \generator demonstrates a better performance in terms of satisfying the properties at all considered decomposition levels. 

\vspace{-1mm}
\subsection{Intuition behind \generator}
\vspace{-1mm}

The main idea behind our \generator is to take the \interactions in decomposed graphs into consideration. Recall that the null-model without such consideration in Sect.~\ref{sec:exp} is shattered into isolated cliques without a giant connected component once it is decomposed into higher decomposition levels.

Intuitively, in order to reproduce the desired patterns in multi-level decomposed graphs, the generation process should have the following characteristics:
\bit
    \item For heavy-tailed degree distribution, ``the rich should get richer'' \cite{barabasi1999emergence}. However, in order to recapture such pattern at higher decomposition levels, groups of nodes should ``get rich'' together rather than individually.
    \item In order to lead to a high clustering coefficient, communities of correlated nodes should form. As an analogy, in research publications, authors tend to collaborate with those who are on the same field or affiliation, rather than any authors. 
    \item However, several pairs of nodes among the communities should also be connected in order for the graph to have a giant connected component and a small effective diameter. 
\eit

\vspace{-1mm}
\subsection{Details of \generator}
\label{sec:hyperPA}


We describe our proposed generator \generator, whose pseudocode is provided in Algorithm \ref{algo:main}. \generator repeatedly introduces a new node to the hypergraph, and form\blue{s} new hyperedges. When a node is added, \generator creates $k$ new hyperedges where $k$ is sampled from a predetermined distribution $NP$. For each new hyperedge introduced by this new node, its size $s$ is sampled from a predetermined distribution $S$. When choosing other nodes to fill in this new hyperedge, it takes into consideration all groups containing $s-1$ nodes. Among all such groups, the chance of being chosen for each group is proportional to its degree. 
The \textit{degree} of each group is defined as the number of hyperedges containing that group.

\generator uses 3 statistics: the number of nodes $n$, the distribution of hyperedge sizes $S$ and the distribution of the number of new hyperedges per new node $NP$.
We obtain them from the real dataset whose patterns \generator is trying to reproduce. Regarding $NP$, we sort hyperedges according to timestamps, and reassign nodes into new node ids based on this chronological order. We then learn NP by accounting, for each (new) node id $i$, $HE_i - HE_{i-1}$, where $HE_i$ is the number of hyperedges consisting of nodes with ids less than or equal to $i$.


\begin{algorithm}[t]
	\small
	\caption{\generator: Hypergraph generator based on Preferential Attachment (Proposed Model) \label{algo:main}} 
	\SetKwInOut{Inputs}{Inputs}
	\SetKwInOut{Output}{Output}
	\Inputs{(1) distribution of hyperedge sizes $S$ (with max size $\bar{s}$),\\ (2) distribution of number of new hyperedges $NP$,\\ (3) number of nodes $n$}
	\Output{synthetic hypergraph $G$}
	initialize $G$ with $\floor{{\bar{s}}/{2}}$ disjoint hyperedges of size $2$, and compute the degree of all their subsets\\
	\For{$i \gets 1$ to $n$}{
		sample a number $k$ from $NP$.\\
		\For{$j \gets 1$ to $k$}{
			sample a hyperedge size s from $S$ \\
			\uIf{$s=1$}{
				add the hyperedge $\{i\}$ to $G$
			}
			\ElseIf{\normalfont{all $(s-1)$-sized groups have 0 degree}}{
				choose $s-1$ nodes randomly\\
				add the hyperedge of $i$ and the $s-1$ nodes to $G$
			}\Else{
				choose a group of size $(s-1)$ with probability proportional to degree \\
				add the hyperedge of $i$ and the $s-1$ nodes to $G$
			}
		}
		\For{{\bf each} \normalfont{of the $k$ newly formed hyperedges with $i$}}{
			increase the degree of all its subsets by $1$
		}
	}
\end{algorithm}

In Algorithm \ref{algo:main}, most of the times when $s>1$, lines 12-13 are executed (a proof is given in Appendix~\ref{appendix:proof_hyperPA}), where \generator chooses a group of nodes based on its degree. As preferential attachment is conducted in a group-like manner, nodes ``get rich'' together, and when decomposed, they form communities, leading to \blue{a} high clustering coefficient.
When a new node is introduced, it forms multiple hyperedges. Since these hyperedges involve nodes from different communities, the introduction of a new node can potentially connect several communities, leading to a giant connected component and a small effective diameter.

 
\generator preserves \interactions, in the sense that most of the times, all of the nodes chosen to fill in a new hyperedge are those from the same previous hyperedge. In order to compare against \generator, we examine two baseline models, \textit{NaivePA} and {\it \ssl}, in the following subsections. They exhibit no or weak \interactions, respectively. 

\vspace{-1mm}
\subsection{Baseline models}
\vspace{-1mm}

\subsubsection{Baseline preferential attachment for hypergraphs}
We consider a naive extension of preferential attachment to hypergraphs. In this model, when filling in each hyperedge of each new node, existing nodes are chosen independently with a chance proportional to their individual degrees (instead of choosing groups of nodes based on degrees of groups). 
We refer to this model as \textit{NaivePA}. Its pseudocode is provided in Appendix~\ref{appendix:pseudocode_naivePA}.

\subsubsection{Subset Sampling}
This model, namely {\it Subset Sampling}, is inspired by \textit{Correlated Repeated Unions} \cite{benson2018sequences}, which was introduced to recapture temporal patterns in hyperedges. 
In \ssl, when a new hyperedge is formed, previous hyperedges are sampled, and then with a certain probability, their elements are chosen independently to fill in the new hyperedge. 
The pseudocode and details of \ssl can be found in Appendix~\ref{appendix:pseudocode}.


\ssl preserves \interactions to some degree, as some nodes in the same previous hyperedge can co-appear in the new hyperedge.
However, as demonstrated in Table \ref{tab:summary_connect_diameter_clust}, the \interactions captured by \ssl are often not connected well enough, making decomposed graphs easily shattered into isolated cliques without retaining a giant connected component.

\vspace{-1mm}
\subsection{Empirical evaluation}
\label{sec:gen:exp}
\vspace{-1mm}


We empirically investigate the properties of generated hypergraphs at four levels of decomposition. To facilitate comprehensive evaluation, we consider the following four datasets, which exhibit the 20 patterns most clearly (4 decomposed graphs $\times$ 5 patterns) to test the three generators on: 
\textit{DAWN}, \textit{email-Eu}, \textit{tags-ask-ubuntu}, and \textit{tags-math}.  
The generators are evaluated on how well they can reproduce the patterns in the real datasets.

We applied the proposed and baseline hypergraph generators to reproduce the real-world hypergraphs. For each considered real hypergraph, the distribution $S$ of the sizes of hyperedges, the distribution $NP$ of the number of new hyperedges per new node, and the exact number $n$ of nodes are directly learned.
Note that $S$, $NP$ and $n$ are the control variables exclusive to hypergraphs that are not directly relevant to how groups of nodes interact with each other, and thus they are out of the scope of this research.

In this paper, we make use of these variables learned directly from the real hypergraphs. Thus, for each real dataset, there are 3 corresponding synthetic datasets, generated by \generator, \ssl and NaivePA using the statistics $S$, $NP$ and $n$ obtained from the real dataset. Generating hypergraphs without explicitly accounting for these 3 variables is left as a topic for future research.

We measure the statistics from the decomposed graphs of the generated hypergraphs and calculate \blue{the} scores \blue{for the 3 generators}:
\begin{itemize}[$\bullet$]
\item \textbf{P1. Giant Conn. Comp.}: if the decomposed graph at that level of the generated hypergraph retains a giant connected component (as described in Sect.~\ref{sec:pattern:giant}), 1 point is given.
\item \textbf{P2. Heavy-tailed Degree Dist.}: the similarity between the generated degree distribution and the real distribution is measured by the Kolmogorov-Smirnov D-statistic, defined as $\max_{x}\{\|F'(x)-F(x)\|\}$ where $F, F'$ are the cumulative degree distributions of the corresponding real and generated decomposed graphs. 1 point is given to the generator having the D-statistic smaller than $0.2$.
\item \textbf{P3. Small Diameter}: we want the generated effective diameter $d'$ to be close to the real value $d$. As the pattern \textbf{P3.} is `small effective diameter', $d'$ should not be too large. At the same time, $d'$ being too small may be the sign of the decomposed graph being shattered without a `giant connected component'. We adopt a heuristic of the acceptance range as $(\frac{2d}{3}, \frac{4d}{3})$. If $d'$ is in the acceptance range, 1 point is given.
\item \textbf{P4. High Clustering Coeff.}: as the pattern \textbf{P4.} is `high clustering coefficient', it is desirable for the generated clustering coefficient $c'$ not to be too small compared to the real value $c$. However, $c'$ being too large may imply that the graph is shattered into isolated cliques. As the clustering coefficient is bounded above by 1, we adopt a heuristic of the acceptance range as $(\frac{2c}{3}, \min(\frac{4c}{3}, 1))$. If $c'$ is in the acceptance range, 1 point is given.
\item \textbf{P5. Skewed Singular Val.}: similar to \textbf{P2.}, the similarity between the singular-value distributions of the real and generated datasets is measured by the Kolmogorov-Smirnov D-statistic. 1 point is given to the generator having the D-statistic smaller than $0.2$.
\end{itemize}

Results of the generators are compared visually in Fig.~\ref{fig:summary_tags-ask-ubuntu_figure} and numerically in Tables ~\ref{tab:summary_Dstats} and \ref{tab:summary_connect_diameter_clust}. The total scores from the two tables for \generator, NaivePA and \ssl are 64, 49 and 57, respectively.
Note that our proposed model, \generator achieved the highest score. 
Without accounting for \interactions, variables $S$, $NP$ and $n$ are not sufficient to reproduce the patterns, as NaivePA and \ssl fail to do so even when utilizing $S$, $NP$ and $n$. 

\input{091summary_Dstats.tex}
\input{091summary_connect_diameter_clust.tex}

%% file: 091summary_Dstats.tex
\begin{table}[t!]
\begin{center}
\vspace{-4mm}
\caption{\label{tab:summary_Dstats}D-statistics between the distributions of real and synthetic datasets generated by the 3 models. We generated each dataset 5 times and report the average. 1 point is given for each D-statistic smaller than $0.2$ and the total scores are computed at the end. \generator achieved the highest score.}
\scalebox{0.85}{
\begin{tabular}{l|l|ccc}
\toprule
Dataset & Level  & {\bf \generator}     & Naive         & Subset\\ 
        &        & {\bf (Proposed)}     & PA            & Sampling \\
\midrule  
\multicolumn{4}{l}{\textbf{Degree distribution}} \\ 
\midrule

DAWN                &  Node        & 0.153 & 0.184 & 0.132  \\
					& Edge         & 0.135  & 0.082 & 0.059  \\ 
					& Triangle     & 0.117 & 0.077 & 0.203  \\ 
					& 4clique      & 0.048 & 0.041 & 0.049  \\ 

\midrule
email-Eu            &  Node        & \red{0.392} & \red{0.282} & \red{0.235} \\
					&  Edge        & 0.109 & 0.148 & 0.126  \\ 
					&  Triangle    & 0.159 & 0.19 & 0.178  \\ 
					&  4clique     & 0.128 & 0.149 & 0.141   \\ 
\midrule
tags-ask-ubuntu     &  Node        & 0.065 & \red{0.259} & 0.128  \\
					&  Edge        & 0.082 & \red{0.232} & 0.057  \\ 
					&  Triangle    & 0.069 & \red{0.428} & 0.049  \\ 
					&  4clique     & 0.087 & \red{0.655} & 0.029  \\ 
\midrule
tags-math           &  Node         & 0.2 & \red{0.364} & \red{0.249}  \\
					&  Edge         & 0.101 & \red{0.216} & 0.073  \\ 
					&  Triangle     & 0.072 & \red{0.365} & 0.117  \\ 
					&  4clique      & 0.025 & \red{0.615} & 0.077  \\ 

      \midrule
      \multicolumn{4}{l}{\textbf{Singular-value distribution}} \\ 
      \midrule
DAWN                 &  Node        & 0.2  & 0.162 &  0.125   \\
					& Edge         & 0.167 & \red{0.227} & \red{0.259}  \\ 
					& Triangle     & \red{0.256} & \red{0.21} & \red{0.335}  \\ 
					& 4clique      & \red{263} & \red{0.37} & \red{0.433}  \\ 

\midrule
email-Eu             &  Node        & \red{0.413} & 0.185 &  0.2 \\
					&  Edge        &  0.185 &  \red{0.223} & \red{0.216} \\ 
					&  Triangle    & \red{0.219} & \red{0.376} & \red{0.497}  \\ 
					&  4clique     & \red{0.408} & \red{0.488} & \red{0.407}   \\ 
\midrule
tags-ask-ubuntu      &  Node        & \red{0.226} & \red{0.21} & \red{0.225} \\
					&  Edge        & 0.169 & \red{0.397} & \red{0.322}  \\ 
					&  Triangle    & \red{0.288} & \red{0.373} & \red{0.369}  \\ 
					&  4clique     & \red{0.215} & \red{0.507} & \red{0.521}  \\ 
\midrule
tags-math           &  Node         & \red{0.228} & 0.168 & \red{0.502}  \\
					&  Edge         & \red{0.241} & \red{0.348} & 0.116  \\ 
					&  Triangle     & \red{0.344} & \red{0.491} & \red{0.292}  \\ 
					&  4clique      & \red{0.3} & \red{0.51} & \red{0.369}  \\ 
      \midrule
      \midrule
      \multicolumn{2}{l|}{\textbf{Score}} & 19 & 10 & 17 \\ 
      \bottomrule
\end{tabular} 
}
\end{center}
\end{table}
\vspace{-1mm}

%% file: 091summary_connect_diameter_clust.tex
\begin{table}[t!]
\begin{center}
\vspace{-4mm}
\caption{\label{tab:summary_connect_diameter_clust}Graph statistics of real and synthetic datasets at all 4 decomposition levels. The scores for the generators are listed at the end. \generator achieved the highest score.}
\scalebox{0.85}{
\begin{tabular}{l|l|c||ccc}
\toprule
Dataset & Level  & Real      & {\bf \generator}     & Naive         & Subset\\ 
        &        & Data      & {\bf (Proposed)}     & PA            & Sampling \\
\midrule  
\multicolumn{4}{l}{\textbf{Connected component}} \\ 
\midrule
DAWN                 &  Node    & 0.89   & 0.996 & 0.73 & 0.999  \\
	                 & Edge     & 0.98   & 0.98  & 0.95 & 0.95  \\ 
                     & Triangle & 0.91   & 0.89  & \red{0.08} & 0.79  \\ 
                     & 4clique  & 0.52   & 0.81  & \red{0.01} & 0.22  \\ 
      
      \midrule
email-Eu             &  Node    & 0.98   & 0.995 & 0.997 & 0.988 \\
	                 &  Edge    & 0.98   & 0.86 & 0.935 & 0.8  \\ 
                     &  Triangle& 0.86   & 0.86 & 0.54 & 0.5  \\ 
                     &  4clique & 0.41   & 0.76 & \red{0.03} & \red{0.04}   \\ 
      \midrule
tags-ask-ubuntu      &  Node    & 0.99   & 0.99 & 0.99 & 0.99  \\
	                 &  Edge    & 0.98   & 0.92 & 0.98 & 0.95  \\ 
                     &  Triangle& 0.79   & 0.81 & 0.74 & 0.55  \\ 
                     &  4clique & 0.21   & 0.39 & 0.11 & \red{0.002}  \\ 
  \midrule
tags-math           &  Node     & 0.99   & 0.997 & 0.997 & 0.996  \\
	                &  Edge     & 0.99   & 0.98 & 0.993 & 0.97  \\ 
                    &  Triangle & 0.91    & 0.81 & 0.77 & 0.55  \\ 
                    &  4clique  & 0.35    & 0.28 & 0.12 & \red{0.02}  \\ 
\midrule  
\multicolumn{4}{l}{\textbf{Diameter}} \\ 
\midrule

DAWN                 &  Node    & 2.6   & 2    & 1.84    & 2  \\
	                 & Edge     & 3.9   & 3.5  & \red{6.8}     & 3.9  \\ 
                     & Triangle & 5.3   & 3.9  & \red{11.2}    & 5.9  \\ 
                     & 4clique  & 8.1   & 5.5  & 9.9     & 8.26  \\ 
      
      \midrule
email-Eu             &  Node    & 2.8   & 1.96 & 1.93    & 1.96 \\
	                 &  Edge    & 5.7   & \red{3.4}  & 4.4     & 4.8  \\ 
                     &  Triangle& 10.3  & \red{3.9}  & 6.4     & 6.9  \\ 
                     &  4clique & 15.3    & \red{6.9}  & \red{9.15}    & \red{6.5}  \\ 
      \midrule
tags-ask-ubuntu      &  Node    & 2.4   & 1.95 & 1.9     & 1.95  \\
	                 &  Edge    & 4.5   & 4.4  & 3.8     & 4.6  \\ 
                     &  Triangle& 7.8   & 7    & 5.77    & 8.2  \\ 
                     &  4clique & 17.1  & 15.75& \red{9.1}     & \red{5.8}  \\ 
  \midrule
tags-math           &  Node     & 2.1   & 1.9  & 1.88    & 1.9  \\
	                &  Edge     & 3.9   & 4.4  & 3.76    & 4.5  \\ 
                    &  Triangle & 6.7   & 8.2  & 5.75    & 7.5  \\ 
                    &  4clique  & 14.8  & 18.9 & \red{8.5}     & \red{8}  \\ 
      \midrule
      \multicolumn{4}{l}{\textbf{Clustering coefficient}} \\ 
      \midrule
DAWN                 &  Node    & 0.64   & 0.82  & \red{0.37} & 0.78  \\
	                 & Edge     & 0.72   & 0.76  & 0.82 & 0.7  \\ 
                     & Triangle & 0.87   & 0.77  & 0.96 & 0.86  \\ 
                     & 4clique  & 0.89   & 0.85  & 0.62 & 0.73  \\ 
      
      \midrule
email-Eu             &  Node    & 0.49   & 0.81 & 0.73 & 0.63 \\
	                 &  Edge    & 0.81   & 0.68 & 0.78 & 0.71  \\ 
                     &  Triangle& 0.89   & 0.8  & 0.85 & 0.89  \\ 
                     &  4clique & 0.89   & 0.9  & 0.6  & 0.66   \\ 
      \midrule
tags-ask-ubuntu      &  Node    & 0.61    & 0.6  & 0.72 & 0.62  \\
	                 &  Edge    & 0.75   & 0.71 & 0.76 & 0.74  \\ 
                     &  Triangle& 0.89  & 0.74 & 0.9  & 0.83  \\ 
                     &  4clique & 0.74   & 0.69 & 0.67 & \red{0.34}  \\ 
  \midrule
tags-math           &  Node     & 0.63   & 0.67 & 0.73 & 0.65  \\
	                &  Edge     & 0.71   & 0.68 & 0.69 & 0.7  \\ 
                    &  Triangle & 0.85   & 0.75 & 0.9  & 0.825  \\ 
                    &  4clique  & 0.71    & 0.67 & 0.68 & \red{0.33}  \\ 
      \midrule
      \midrule
      \multicolumn{3}{l|}{\textbf{Score}} & 45 & 39 & 40 \\ 
      \bottomrule
      \multicolumn{6}{l}{~}  \\ 
\end{tabular} 
}
\end{center}
\vspace{1mm}
\end{table}

%% file: 060conclusion.tex
In summary, our contributions in this work are threefold.

\smallsection{Multi-level decomposition:}
First, we propose the multi-level decomposition as an effective means of investigating hypergraphs. The multi-level decomposition has several benefits: (1) it captures the group interactions within the hypergraph, (2) its graphical representation provides convenience in leveraging existing tools, and (3) it represents the original hypergraph without information loss.  

\smallsection{Patterns in real hypergraphs:}
Then, we present a set of common patterns held in $13$ real-world hypergraphs.
Specifically, we observe the following structural properties consistently at different decomposition levels 
(1) {\it giant connected components}, (2) {\it heavy-tailed degree distributions}, (3) {\it small effective diameters}, (4) {\it high clustering coefficients}, and (5) {\it skewed singular-value distributions}.
The connectivity of \interactions, however, varies among domains of datasets, illustrated by the level of decomposition that shatters the dataset into small connected components.

\smallsection{Realistic hypergraph generator:}
Lastly, we introduce \generator, a  hypergraph generator that is simple but capable of reproducing the patterns of real-world hypergraphs across different decomposition levels. By maintaining the connectivity of \interactions of nodes in the hypergraphs,  \generator shows better performance in reproducing the patterns than two other baseline models.

\smallsection{\bf Reproducibility:} We made the datasets, the code, and the full experimental results available at \url{https://github.com/manhtuando97/KDD-20-Hypergraph}.

%% file: 070appendix_reproducibility.tex

\vspace{-1mm}
\section{Appendix: Dataset description}
\label{appendix:description}
\vspace{-1mm}

The thirteen datasets investigated in our work  are from the following sources:
\begin{itemize}[$\bullet$]
	\item \textbf{Publication coauthors}: each node is an author and each hyperedge is a publication involving one or several coauthors. The coauthorship hypergraphs considered in this paper are \textit{coauth-DBLP} \footnote{https://dblp.org/xml/release/}, \textit{coauth-Geology} \cite{sinha2015MAG}, \textit{coauth-History} \cite{xu2013hyperlink}.
	\item \textbf{Drug abuse warning network (\textit{DAWN}) drugs}: each node is a drug and each hyperedge is a set of drugs used by a patient.
	\item \textbf{Emails from an European research institution (\textit{email-Eu})}: each node is an email address and each hyperedge is a set of sender and all recipients of an email \cite{yin2017local}.
	\item \textbf{National drug code directory (NDC) drugs}: each node is a class label (in \textit{NDC-classes}) or a substance (in \textit{NDC-substances}) and each hyperedge is the set of labels/substances of a drug.
	\item \textbf{Online question tags}: each node is a tag and each hyperedge is the set of tags attached to a question in an online forum. We considered \textit{tags-ask-ubuntu}\footnote{https://askubuntu.com/}, \textit{tags-math}\footnote{https://math.stackexchange.com/}, \textit{tags-stack-overflow}\footnote{https://stackoverflow.com/}.
	\item \textbf{Thread participants}: each node is a user registered in an online forum and each hyperedge is the set of users participating in a discussion thread. There are 3 considered datasets: \textit{threads-ask-ubuntu}, \textit{threads-math}, \textit{threads-stack-overflow}.
\end{itemize}
We thank the authors of \cite{benson2018simplicial} for making the datasets publicly available for our research purposes. From the raw format, we preprocess each hypergraph to retain only unique hyperedges since duplicated hyperedges do not contribute to the above patterns. The distribution of hyperedge sizes are ploted in Fig.~\ref{fig:size_distribution}. For the decomposed graphs of the datasets, the numbers of nodes are reported in Tables~\ref{tab:node_properties} and \ref{tab:edge_triangle_4clique_properties}, and the numbers of edges are listed in Table~\ref{tab:summary_number_edges}. 

\input{070summary_number_edges.tex}

\begin{figure}[ht]
	\centering
	\vspace{0mm}
	\scalebox{0.5}{
	\includegraphics{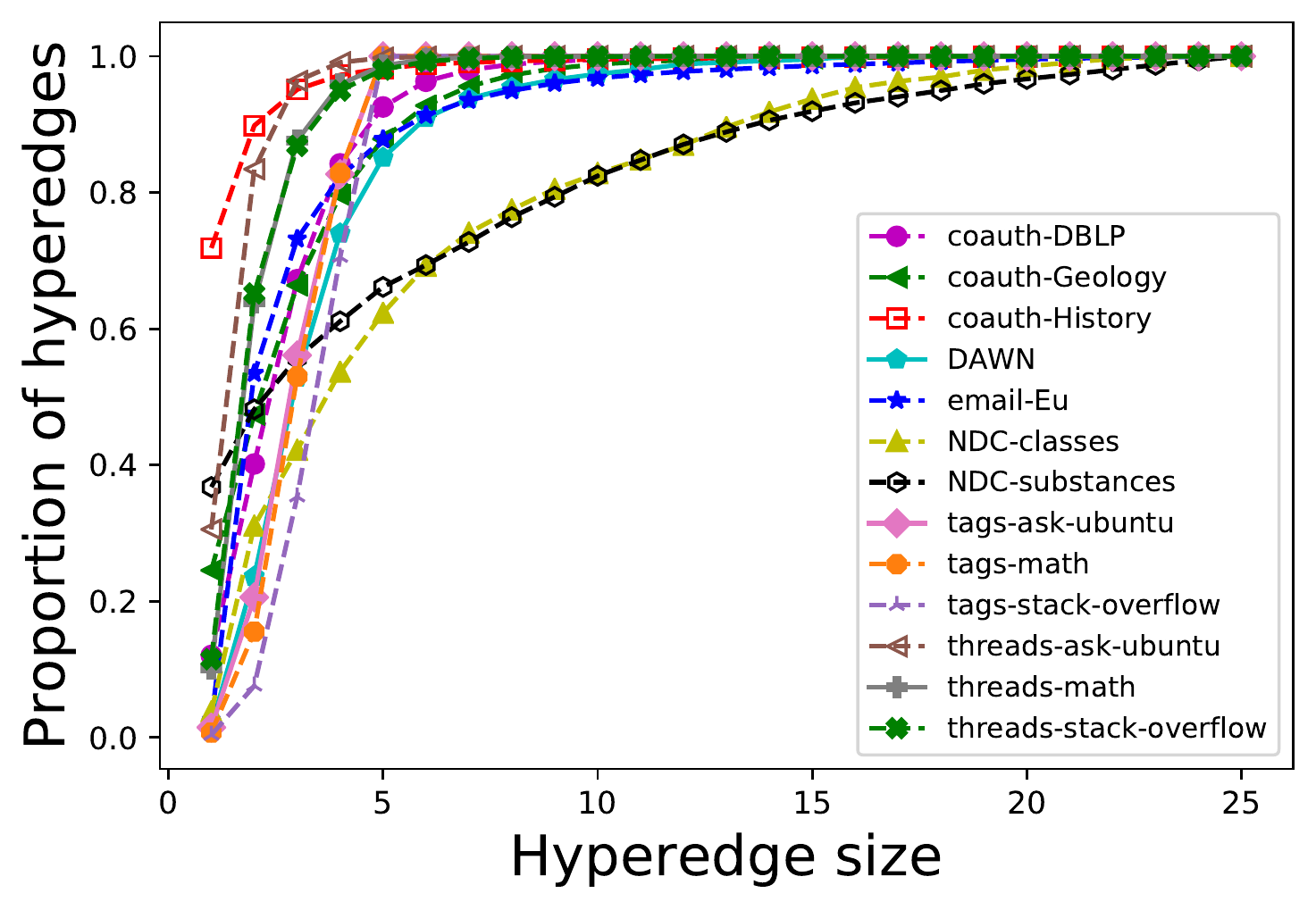}
	}
	\vspace{0mm}
	\caption{\label{fig:size_distribution}
		Cumulative distribution of hyperedge sizes.
	}
\end{figure}


\begin{algorithm}[t]
	\small
	\caption{NaivePA (Baseline Model) \label{algo:naive}} 
	\SetKwInOut{Inputs}{Inputs}
	\SetKwInOut{Output}{Output}
	\Inputs{(1) distribution of hyperedge sizes $S$ (with max size $\bar{s}$),\\ (2) distribution of number of new hyperedges $NP$,\\ (3) number of nodes $n$}
	\Output{synthetic hypergraph $G$}
	initialize $G$ with $\floor{{\bar{s}}/{2}}$ disjoint hyperedges of size $2$, and compute the degree of all nodes in them \\
	\For{$i \gets 1$ to $n$}{
		sample a number $k$ from $NP$.\\
		\For{$j \gets 1$ to $k$}{
			sample a hyperedge size s from $S$ \\
			\uIf{$s=1$}{
				add the hyperedge $\{i\}$ to $G$

			}\Else{
				choose $(s-1)$ nodes independently with probability proportional to their degrees \\
				add the hyperedge of $i$ and the $s-1$ nodes to $G$
			}
		}
		
		\For{{\bf each} \normalfont{node involved the $k$ newly formed hyperedges with $i$}}{
			increase the degree of each node by the number of its involving hyperedges
		}
	}
\end{algorithm}

\vspace{-1mm}
\section{Appendix: Pseudocode}

\vspace{-1mm}
\subsection{Pseudocode for NaivePA}
\vspace{-1mm}
\label{appendix:pseudocode_naivePA}
Pseudocode for NaivePA is provided in Algorithm~\ref{algo:naive}. Unlike \generator, which maintains the degree of every subset of every hyperedge, NaivePA only maintains the degree of individual nodes. When forming hyperedges with each newly arrived node, NaivePA chooses several nodes independently based on their degrees.

\begin{algorithm}[t]
	\small
	\caption{Subset Sampling\label{algo:SS} (Baseline Model)}
	\SetKwInOut{Inputs}{Inputs}
	\SetKwInOut{Output}{Output}
	\Inputs{(1) distribution of hyperedge sizes $S$, \\ (2) distribution of number of new hyperedges $NP$, \\ (3) number of nodes $n$, \\ (4) sampling rule $P$, \\ (5) probability $p$}
	\Output{synthetic hypergraph $G$}
	initialize $G$ with $2$ disjoint hyperedges of maximum size in $S$ \\
	\For{$i \gets 1$ to $n$}{
		sample a number $k$ from $NP$\\
		\For{$j \gets 1$ to $k$}{
			sample a size s from $S$ \\
			\eIf{$s = 1$}{
				add the hyperedge $\{i\}$ to $G$
			}{
				initialize $B$ to $\{i\}$\\
				\While{$|B| < s$}{
					initialize $T$ to an empty set\\
					sample a hyperedge $E$ of $G$ based on $P$\\
					sample each node $v \in E$ into $B$ with prob. $p$ \\
					\eIf{$|T| \leq s - |B|$}{
						$B \gets B \cup T$\\
					}{
						$T\leftarrow$ randomly chosen $s-|B|$ nodes in $T$ \\
						$B \gets B \cup T$ 
					}
				} 
				add the hyperedge $B$ to $G$
			}
		} 
	}
\end{algorithm}

\vspace{-1mm}
\subsection{Pseudocode for \ssl}
\label{appendix:pseudocode}
\vspace{-1mm}

We present the pseudocode for \ssl in Algorithm \ref{algo:SS}.
For \ssl, in order to keep the model simple, we tried the following variants for the sampling rule $P$: 
\bit
\item{\textit{random}}: a hyperedge is randomly chosen among all previously formed hyperedges.
\item{\textit{recent}}: among all available hyperedges $E_1, E_2,...,E_n$, hyperedge $E_i$ has probability of being chosen equal to $\frac{i}{\sum_{j=1}^{n} j}$.
\item{\textit{$k$ most recent}}: only sample a set based on \textit{random} or \textit{recent} from the  $k$ most recent hyperedges.
\eit
Empirical data shows that when $P$ is \textit{k most recent}, the resulting graph has an unrealistically high diameter, while none between \textit{random} and \textit{recent} outperforms the other. For probability $p$, increasing from $0.4$ does not significantly change the result, while too low values  make the graph shattered at the triangle-level decomposition. The reported results of \textit{Subset Sampling} are from $p = 0.8$ and $P= random$.

\vspace{-1mm}
\section{Appendix: Proofs}
\label{appendix:proof}

\vspace{-1mm}
\subsection{Recovering hypergraphs from decomposed graphs}
\label{appendix:proof:recovery}
\vspace{-1mm}

In this section, we prove that the original hypergraph can be recovered exactly from its decomposed graphs.
To this end, we consider decomposed graphs with self-loops and edge weights, which are ignored in the previous sections since they do not contribute to the presented patterns.
Specifically, for each $k$-level decomposed graph $G_{(k)}=(V_{(k)},E_{(k)})$ of a hypergraph $G=(V,E)$, we introduce a weight function $\omega_{(k)}$, defined as follows:
\begin{align*}
& \omega_{(k)} \big(\{u_{(k)},v_{(k)}\} \big) :=
|\{e\in E: u_{(k)} \cup v_{(k)}\subseteq e\}|.
\end{align*}
That is, for each edge $\{u_{(k)},v_{(k)}\}$ in $E_{(k)}$, $\omega_{(k)}$ gives the number of hyperedges in $E$ that contain the union of $u_{(k)}$ and $v_{(k)}$.
Additionally, for each hyperedge $\{a\}\in E$ of size $1$, we add a self-loop to the node $\{a\}$ in the $1$-level decomposed graph.

\begin{theorem}\textsc{(Recovery of original hypergraphs).}
Assume that the maximum size of a hyperedge in a given hypergraph is $m$. If we have all the decomposed graphs up to level $(m-1)$ with edge weights and self-loops, we can recover the original hypergraph.
\end{theorem}

\begin{proof}

Initialize an empty set $S = \emptyset$, which will contain the recovered hyperedges.
We recover the hyperedges sized from the largest to smallest.
By our definition, a hyperedge of size $n > k$ results in a clique of size ${n}\choose{k}$ in the $k$-level decomposed graph.

We start with the $(m-1)$-level decomposed graph: for each edge between two distinct $(m-1)$-level nodes $\{a_1,...,a_{m-1}\}$ and $\{b_1, ...,b_{m-1}\}$, as $m$ is the maximum size for any hyperedge, the union of these two $(m-1)$-level nodes must be an original hyperedge $e$ of size $m$. We add this hyperedge $e$ into $S$ and decrement the weight of each edge involved in the resulting clique of $e$ in the $(m-1)$-level decomposed graph. We keep doing this until we completely clear the graph 
(i.e., making the weights of all edges to $0$) to recover all hyperedges of size $m$. 

Assume that we have recovered all hyperedges of sizes greater than $k$ and have stored them in $S$. In the $(k-1)$-level decomposed graph, we decrement the weight of each edge involved in the clique resulting from each hyperedge currently in $S$.
Then, we repeat the process above to recover all hyperedges of size $k$.

By continuing this procedure, eventually we can also recover all hyperedges of sizes at least $2$ after processing the node-level decomposed graph (i.e., $1$-level decomposed graph).
Since we also maintain self-loops, we can recover all hyperedges of size $1$.
The proof is completed here.
\end{proof}

\vspace{-1mm}
\subsection{Randomness in \generator}
\label{appendix:proof_hyperPA}
\vspace{-1mm}

We present a simple proof about the randomness in \generator. 

\begin{theorem}\textsc{(Randomness in \generator).}
Given that the largest size $s$ possible in the distribution $S$ is a finite number $\bar{s}$, the conditional statement at line 8 of Algorithm~\ref{algo:main}, denoted as statement \textbf{U}, holds at most $\floor{\frac{\bar{s}}{2}-1}$ times.
\end{theorem}

\begin{proof}
Assume that at a given time step $t$, the sampled size at line 5 is $s$ and \textbf{U} holds. Then, the following conditions must be satisfied:
\begin{enumerate}
\item All $(s-1)$-sized groups have 0 degree, i.e., up to the time step $t$, only hyperedges of sizes up to $s-2$ present in the hypergraph,
\item $s \geq 4 $,
\end{enumerate}
where the second condition is from the first condition and the fact that the hypergraph is initialized with $12$ hyperedges of size $2$.
Denote two consecutive time steps when \textbf{U} holds as $t$ and $t'$, respectively. 
Denote the hyperedge sizes sampled at line 5 at time steps $t$ and $t'$ as $s$ and $s'$, respectively. 
According to the above two conditions, $s \geq 4$ and $s \leq s' - 2$. 
Assume \textbf{U} holds $M$ times at time steps $t_{1}, .., t_{M}$, and denote the hyperedge sizes sampled at line 5 of the algorithm at these time steps as $s_{1}, ..., s_{M}$, respectively. Then, as shown, $$s_{1} \leq s_{2} - 2 \leq s_{3} - 4 \leq ... \leq s_{M} - 2 \times (M-1).$$
Then, $2 \times (M-1) \leq s_{M} - s_{1}$.
This, $s_{1} \geq 4$, and $s_{M} \leq \bar{s}$ imply $2 \times (M-1) \leq \bar{s} - 4$ or equivalently $M \leq \frac{\bar{s}}{2}-1$.
As $M$ must be an integer, we conclude that $M \leq \floor{\frac{\bar{s}}{2}-1}$.
\end{proof}

As in our datasets, the maximum hyperedge size is 25 and the distribution $S$ used for \generator is learned from the dataset, we have $\bar{s} = 25$ for \generator. According to the proof, the conditional statement at line 8 of Algorithm~\ref{algo:main} can only hold at most 11 times.
If the number of nodes $n$ is relatively large, most of the time when $s>1$, the conditional statement at line 8 in Algorithm~\ref{algo:main} does not hold, indicating that lines 12-13 of the pseudocode are executed.

%% file: 070summary_number_edges.tex
\begin{table}[ht!]
	\centering
	\caption{\label{tab:summary_number_edges}
		Number of edges in the decomposed graphs.
	}
	\scalebox{0.86}{
		\begin{tabular}{l|cccc}
			\toprule
			Dataset           & $|E_{(1)}|$ & $|E_{(2)}|$ & $|E_{(3)}|$ & $|E_{(4)}|$    \\
			\midrule
			
			coauth-DBLP        & 7,904,336  &  31,284,160     & 50,887,503 & 35,299,764  \\
			coauth-Geology     & 5,120,762 &  18,987,747    & 35,384,178  & 26,839,940  \\
			coauth-History     &  1,156,914  &  1,852,269 & 3,001,774  & 2,183,900 \\
			DAWN                   & 122,963    &   1,682,274    &  4,097,770 & 3,219,360  \\
			email-Eu               &   29,299   &   155,769    &  393,527 & 360,955  \\
			NDC-classes            &   6,222    &    20,568   & 45,793 & 38,525  \\
			NDC-substances         &  88,268   & 116,967     &  268,057 & 231,445 \\
			tags-ask-ubuntu        & 132,703   &   1,275,135    & 1,256,181  &  254,750 \\
			tags-math              &   91,685 &   1,217,031    &  1,375,434 &  292,440\\
			tags-stack-overflow    &  4,147,302 &  57,815,235    &  71,817,873 &16,327,590 \\
			threads-ask-ubuntu     &  187,157  &   227,547    & 175,627  & 85,665  \\
			threads-math           &  1,089,307 &   2,810,934    &   3,086,411 & 1,770,730 \\
			threads-stack-overflow & 20,999,838 &   52,797,462  & 66,240,865 & 41,329,315  \\
			\bottomrule
		\end{tabular}
	}
\end{table}

%% file: main.bbl

\begin{thebibliography}{52}


\ifx \showCODEN    \undefined \def \showCODEN     #1{\unskip}     \fi
\ifx \showDOI      \undefined \def \showDOI       #1{#1}\fi
\ifx \showISBNx    \undefined \def \showISBNx     #1{\unskip}     \fi
\ifx \showISBNxiii \undefined \def \showISBNxiii  #1{\unskip}     \fi
\ifx \showISSN     \undefined \def \showISSN      #1{\unskip}     \fi
\ifx \showLCCN     \undefined \def \showLCCN      #1{\unskip}     \fi
\ifx \shownote     \undefined \def \shownote      #1{#1}          \fi
\ifx \showarticletitle \undefined \def \showarticletitle #1{#1}   \fi
\ifx \showURL      \undefined \def \showURL       {\relax}        \fi
\providecommand\bibfield[2]{#2}
\providecommand\bibinfo[2]{#2}
\providecommand\natexlab[1]{#1}
\providecommand\showeprint[2][]{arXiv:#2}

\bibitem[\protect\citeauthoryear{??}{app}{2020}]%
        {appendix}
 \bibinfo{year}{2020}\natexlab{}.
\newblock \bibinfo{title}{Supplementary results, code and datasets}.
\newblock \bibinfo{howpublished}{Available online:
  \url{https://github.com/manhtuando97/KDD-20-Hypergraph}}.
\newblock


\bibitem[\protect\citeauthoryear{Abello, Buchsbaum, and Westbrook}{Abello
  et~al\mbox{.}}{1998}]%
        {abello1998functional}
\bibfield{author}{\bibinfo{person}{James Abello}, \bibinfo{person}{Adam~L
  Buchsbaum}, {and} \bibinfo{person}{Jeffery~R Westbrook}.}
  \bibinfo{year}{1998}\natexlab{}.
\newblock \showarticletitle{A functional approach to external graph
  algorithms}. In \bibinfo{booktitle}{\emph{ESA}}.
\newblock


\bibitem[\protect\citeauthoryear{Abello, Pardalos, and Resende}{Abello
  et~al\mbox{.}}{2013}]%
        {abello2013handbook}
\bibfield{author}{\bibinfo{person}{James Abello}, \bibinfo{person}{Panos~M
  Pardalos}, {and} \bibinfo{person}{Mauricio~GC Resende}.}
  \bibinfo{year}{2013}\natexlab{}.
\newblock \bibinfo{booktitle}{\emph{Handbook of massive data sets}}.
  Vol.~\bibinfo{volume}{4}.
\newblock \bibinfo{publisher}{Springer}.
\newblock


\bibitem[\protect\citeauthoryear{Akoglu, McGlohon, and Faloutsos}{Akoglu
  et~al\mbox{.}}{2008}]%
        {akoglu2008rtm}
\bibfield{author}{\bibinfo{person}{Leman Akoglu}, \bibinfo{person}{Mary
  McGlohon}, {and} \bibinfo{person}{Christos Faloutsos}.}
  \bibinfo{year}{2008}\natexlab{}.
\newblock \showarticletitle{RTM: Laws and a recursive generator for weighted
  time-evolving graphs}. In \bibinfo{booktitle}{\emph{ICDM}}.
\newblock


\bibitem[\protect\citeauthoryear{Akoglu, McGlohon, and Faloutsos}{Akoglu
  et~al\mbox{.}}{2010}]%
        {akoglu2010oddball}
\bibfield{author}{\bibinfo{person}{Leman Akoglu}, \bibinfo{person}{Mary
  McGlohon}, {and} \bibinfo{person}{Christos Faloutsos}.}
  \bibinfo{year}{2010}\natexlab{}.
\newblock \showarticletitle{Oddball: Spotting anomalies in weighted graphs}. In
  \bibinfo{booktitle}{\emph{PAKDD}}.
\newblock


\bibitem[\protect\citeauthoryear{Albert, Jeong, and Barab{\'a}si}{Albert
  et~al\mbox{.}}{1999}]%
        {albert1999internet}
\bibfield{author}{\bibinfo{person}{R{\'e}ka Albert}, \bibinfo{person}{Hawoong
  Jeong}, {and} \bibinfo{person}{Albert-L{\'a}szl{\'o} Barab{\'a}si}.}
  \bibinfo{year}{1999}\natexlab{}.
\newblock \showarticletitle{Internet: Diameter of the world-wide web}.
\newblock \bibinfo{journal}{\emph{Nature}} \bibinfo{volume}{401},
  \bibinfo{number}{6749} (\bibinfo{year}{1999}), \bibinfo{pages}{130}.
\newblock


\bibitem[\protect\citeauthoryear{Albert, Jeong, and Barab{\'a}si}{Albert
  et~al\mbox{.}}{2002}]%
        {albert2002statistical}
\bibfield{author}{\bibinfo{person}{R{\'e}ka Albert}, \bibinfo{person}{Hawoong
  Jeong}, {and} \bibinfo{person}{Albert-L{\'a}szl{\'o} Barab{\'a}si}.}
  \bibinfo{year}{2002}\natexlab{}.
\newblock \showarticletitle{Statistical mechanics of complex networks}.
\newblock \bibinfo{journal}{\emph{Rev. Mod. Phys}} (\bibinfo{year}{2002}).
\newblock


\bibitem[\protect\citeauthoryear{Alstott and Bullmore}{Alstott and
  Bullmore}{2014}]%
        {alstott2014powerlaw}
\bibfield{author}{\bibinfo{person}{Jeff Alstott} {and}
  \bibinfo{person}{Dietmar~Plenz Bullmore}.} \bibinfo{year}{2014}\natexlab{}.
\newblock \showarticletitle{powerlaw: a Python package for analysis of
  heavy-tailed distributions}.
\newblock \bibinfo{journal}{\emph{PloS one}} \bibinfo{volume}{9},
  \bibinfo{number}{1} (\bibinfo{year}{2014}).
\newblock


\bibitem[\protect\citeauthoryear{Barab{\'a}si and Albert}{Barab{\'a}si and
  Albert}{1999}]%
        {barabasi1999emergence}
\bibfield{author}{\bibinfo{person}{Albert-L{\'a}szl{\'o} Barab{\'a}si} {and}
  \bibinfo{person}{R{\'e}ka Albert}.} \bibinfo{year}{1999}\natexlab{}.
\newblock \showarticletitle{Emergence of scaling in random networks}.
\newblock \bibinfo{journal}{\emph{Science}} \bibinfo{volume}{286},
  \bibinfo{number}{5439} (\bibinfo{year}{1999}), \bibinfo{pages}{509--512}.
\newblock


\bibitem[\protect\citeauthoryear{Benson, Abebe, Schaub, Jadbabaie, and
  Kleinberg}{Benson et~al\mbox{.}}{2018a}]%
        {benson2018simplicial}
\bibfield{author}{\bibinfo{person}{Austin~R Benson}, \bibinfo{person}{Rediet
  Abebe}, \bibinfo{person}{Michael~T Schaub}, \bibinfo{person}{Ali Jadbabaie},
  {and} \bibinfo{person}{Jon Kleinberg}.} \bibinfo{year}{2018}\natexlab{a}.
\newblock \showarticletitle{Simplicial closure and higher-order link
  prediction}.
\newblock \bibinfo{journal}{\emph{Proc. Natl. Acad. Sci. U.S.A}}
  \bibinfo{volume}{115}, \bibinfo{number}{48} (\bibinfo{year}{2018}),
  \bibinfo{pages}{E11221--E11230}.
\newblock


\bibitem[\protect\citeauthoryear{Benson, Gleich, and Leskovec}{Benson
  et~al\mbox{.}}{2016}]%
        {benson2016higher}
\bibfield{author}{\bibinfo{person}{Austin~R Benson}, \bibinfo{person}{David~F
  Gleich}, {and} \bibinfo{person}{Jure Leskovec}.}
  \bibinfo{year}{2016}\natexlab{}.
\newblock \showarticletitle{Higher-order organization of complex networks}.
\newblock \bibinfo{journal}{\emph{Science}} \bibinfo{volume}{353},
  \bibinfo{number}{6295} (\bibinfo{year}{2016}), \bibinfo{pages}{163--166}.
\newblock


\bibitem[\protect\citeauthoryear{Benson, Kumar, and Tomkins}{Benson
  et~al\mbox{.}}{2018b}]%
        {benson2018sequences}
\bibfield{author}{\bibinfo{person}{Austin~R Benson}, \bibinfo{person}{Ravi
  Kumar}, {and} \bibinfo{person}{Andrew Tomkins}.}
  \bibinfo{year}{2018}\natexlab{b}.
\newblock \showarticletitle{Sequences of sets}. In
  \bibinfo{booktitle}{\emph{KDD}}.
\newblock


\bibitem[\protect\citeauthoryear{Bollob{\'a}s and Riordan}{Bollob{\'a}s and
  Riordan}{2004}]%
        {bollobas2004diameter}
\bibfield{author}{\bibinfo{person}{B{\'e}la Bollob{\'a}s} {and}
  \bibinfo{person}{Oliver Riordan}.} \bibinfo{year}{2004}\natexlab{}.
\newblock \showarticletitle{The diameter of a scale-free random graph}.
\newblock \bibinfo{journal}{\emph{Combinatorica}} \bibinfo{volume}{24},
  \bibinfo{number}{1} (\bibinfo{year}{2004}), \bibinfo{pages}{5--34}.
\newblock


\bibitem[\protect\citeauthoryear{Bonacich, Holdren, and Johnston}{Bonacich
  et~al\mbox{.}}{2004}]%
        {bonacich2004hyper}
\bibfield{author}{\bibinfo{person}{Phillip Bonacich},
  \bibinfo{person}{Annie~Cody Holdren}, {and} \bibinfo{person}{Michael
  Johnston}.} \bibinfo{year}{2004}\natexlab{}.
\newblock \showarticletitle{Hyper-edges and multi-dimensional centrality}.
\newblock \bibinfo{journal}{\emph{Soc. Netw}} \bibinfo{volume}{26},
  \bibinfo{number}{3} (\bibinfo{year}{2004}), \bibinfo{pages}{189--203}.
\newblock


\bibitem[\protect\citeauthoryear{Broder, Kumar, Maghoul, Raghavan, Rajagopalan,
  Stata, Tomkins, and Wiener}{Broder et~al\mbox{.}}{2000}]%
        {broder2000graph}
\bibfield{author}{\bibinfo{person}{Andrei Broder}, \bibinfo{person}{Ravi
  Kumar}, \bibinfo{person}{Farzin Maghoul}, \bibinfo{person}{Prabhakar
  Raghavan}, \bibinfo{person}{Sridhar Rajagopalan}, \bibinfo{person}{Raymie
  Stata}, \bibinfo{person}{Andrew Tomkins}, {and} \bibinfo{person}{Janet
  Wiener}.} \bibinfo{year}{2000}\natexlab{}.
\newblock \showarticletitle{Graph structure in the web}.
\newblock \bibinfo{journal}{\emph{Computer networks}} \bibinfo{volume}{33},
  \bibinfo{number}{1-6} (\bibinfo{year}{2000}), \bibinfo{pages}{309--320}.
\newblock


\bibitem[\protect\citeauthoryear{C}{C}{2013}]%
        {berge1989hyper}
\bibfield{author}{\bibinfo{person}{Berge C}.} \bibinfo{year}{2013}\natexlab{}.
\newblock \bibinfo{booktitle}{\emph{Hypergraphs}}. Vol.~\bibinfo{volume}{45}.
\newblock \bibinfo{publisher}{North Holland, Amsterdam}.
\newblock


\bibitem[\protect\citeauthoryear{Chodrow}{Chodrow}{2019}]%
        {chodrow2019configuration}
\bibfield{author}{\bibinfo{person}{Philip~S Chodrow}.}
  \bibinfo{year}{2019}\natexlab{}.
\newblock \showarticletitle{Configuration Models of Random Hypergraphs and
  their Applications}.
\newblock \bibinfo{journal}{\emph{arXiv preprint arXiv:1902.09302}}
  (\bibinfo{year}{2019}).
\newblock


\bibitem[\protect\citeauthoryear{Chung and Lu}{Chung and Lu}{2002}]%
        {chung2002average}
\bibfield{author}{\bibinfo{person}{Fan Chung} {and} \bibinfo{person}{Linyuan
  Lu}.} \bibinfo{year}{2002}\natexlab{}.
\newblock \showarticletitle{The average distances in random graphs with given
  expected degrees}.
\newblock \bibinfo{journal}{\emph{Proc. Natl. Acad. Sci. U.S.A}}
  \bibinfo{volume}{99}, \bibinfo{number}{25} (\bibinfo{year}{2002}),
  \bibinfo{pages}{15879--15882}.
\newblock


\bibitem[\protect\citeauthoryear{Clauset, Shalizi, and Newman}{Clauset
  et~al\mbox{.}}{2009}]%
        {clauset2009power}
\bibfield{author}{\bibinfo{person}{Aaron Clauset},
  \bibinfo{person}{Cosma~Rohilla Shalizi}, {and} \bibinfo{person}{Mark~EJ
  Newman}.} \bibinfo{year}{2009}\natexlab{}.
\newblock \showarticletitle{Power-law distributions in empirical data}.
\newblock \bibinfo{journal}{\emph{SIAM review}} \bibinfo{volume}{51},
  \bibinfo{number}{4} (\bibinfo{year}{2009}), \bibinfo{pages}{661--703}.
\newblock


\bibitem[\protect\citeauthoryear{Cooper and Frieze}{Cooper and Frieze}{2003}]%
        {cooper2003general}
\bibfield{author}{\bibinfo{person}{Colin Cooper} {and} \bibinfo{person}{Alan
  Frieze}.} \bibinfo{year}{2003}\natexlab{}.
\newblock \showarticletitle{A general model of web graphs}.
\newblock \bibinfo{journal}{\emph{Random Struct. Algorithms}}
  \bibinfo{volume}{22}, \bibinfo{number}{3} (\bibinfo{year}{2003}),
  \bibinfo{pages}{311--335}.
\newblock


\bibitem[\protect\citeauthoryear{Easley, Kleinberg, et~al\mbox{.}}{Easley
  et~al\mbox{.}}{2010}]%
        {easley2010networks}
\bibfield{author}{\bibinfo{person}{David Easley}, \bibinfo{person}{Jon
  Kleinberg}, {et~al\mbox{.}}} \bibinfo{year}{2010}\natexlab{}.
\newblock \bibinfo{booktitle}{\emph{Networks, crowds, and markets}}.
  Vol.~\bibinfo{volume}{8}.
\newblock \bibinfo{publisher}{Cambridge university press Cambridge}.
\newblock


\bibitem[\protect\citeauthoryear{Edunov, Logothetis, Wang, Ching, and
  Kabiljo}{Edunov et~al\mbox{.}}{2016}]%
        {edunov2016darwini}
\bibfield{author}{\bibinfo{person}{Sergey Edunov}, \bibinfo{person}{Dionysios
  Logothetis}, \bibinfo{person}{Cheng Wang}, \bibinfo{person}{Avery Ching},
  {and} \bibinfo{person}{Maja Kabiljo}.} \bibinfo{year}{2016}\natexlab{}.
\newblock \showarticletitle{Darwini: Generating realistic large-scale social
  graphs}.
\newblock \bibinfo{journal}{\emph{arXiv:1610.00664}} (\bibinfo{year}{2016}).
\newblock


\bibitem[\protect\citeauthoryear{Eikmeier and Gleich}{Eikmeier and
  Gleich}{2017}]%
        {eikmeier2017revisiting}
\bibfield{author}{\bibinfo{person}{Nicole Eikmeier} {and}
  \bibinfo{person}{David~F Gleich}.} \bibinfo{year}{2017}\natexlab{}.
\newblock \showarticletitle{Revisiting power-law distributions in spectra of
  real world networks}. In \bibinfo{booktitle}{\emph{KDD}}.
\newblock


\bibitem[\protect\citeauthoryear{Faloutsos, Faloutsos, and Faloutsos}{Faloutsos
  et~al\mbox{.}}{1999}]%
        {faloutsos1999power}
\bibfield{author}{\bibinfo{person}{Michalis Faloutsos}, \bibinfo{person}{Petros
  Faloutsos}, {and} \bibinfo{person}{Christos Faloutsos}.}
  \bibinfo{year}{1999}\natexlab{}.
\newblock \showarticletitle{On power-law relationships of the internet
  topology}. In \bibinfo{booktitle}{\emph{ACM SIGCOMM computer communication
  review}}, Vol.~\bibinfo{volume}{29}. ACM, \bibinfo{pages}{251--262}.
\newblock


\bibitem[\protect\citeauthoryear{Girvan and Newman}{Girvan and Newman}{2002}]%
        {girvan2002community}
\bibfield{author}{\bibinfo{person}{M. Girvan} {and} \bibinfo{person}{M.~E.~J.
  Newman}.} \bibinfo{year}{2002}\natexlab{}.
\newblock \showarticletitle{Community structure in social and biological
  networks}.
\newblock \bibinfo{journal}{\emph{Proc. Natl. Acad. Sci. U.S.A}}
  \bibinfo{volume}{99} (\bibinfo{year}{2002}).
\newblock


\bibitem[\protect\citeauthoryear{Kang, McGlohon, Akoglu, and Faloutsos}{Kang
  et~al\mbox{.}}{2010}]%
        {kang2010patterns}
\bibfield{author}{\bibinfo{person}{U Kang}, \bibinfo{person}{Mary McGlohon},
  \bibinfo{person}{Leman Akoglu}, {and} \bibinfo{person}{Christos Faloutsos}.}
  \bibinfo{year}{2010}\natexlab{}.
\newblock \showarticletitle{Patterns on the Connected Components of
  TerabyteScale Graphs}. In \bibinfo{booktitle}{\emph{ICDM}}.
\newblock


\bibitem[\protect\citeauthoryear{Kleinberg}{Kleinberg}{2002}]%
        {kleinberg2002small}
\bibfield{author}{\bibinfo{person}{Jon~M Kleinberg}.}
  \bibinfo{year}{2002}\natexlab{}.
\newblock \showarticletitle{Small-world phenomena and the dynamics of
  information}. In \bibinfo{booktitle}{\emph{NIPS}}.
\newblock


\bibitem[\protect\citeauthoryear{Kleinberg, Kumar, Raghavan, Rajagopalan, and
  Tomkins}{Kleinberg et~al\mbox{.}}{1999}]%
        {klein1999the}
\bibfield{author}{\bibinfo{person}{Jon~M Kleinberg}, \bibinfo{person}{Ravi
  Kumar}, \bibinfo{person}{Prabhakar Raghavan}, \bibinfo{person}{Sridhar
  Rajagopalan}, {and} \bibinfo{person}{Andrew~S Tomkins}.}
  \bibinfo{year}{1999}\natexlab{}.
\newblock \showarticletitle{The web as a graph: measurements, models, and
  methods}. In \bibinfo{booktitle}{\emph{COCOON}}.
\newblock


\bibitem[\protect\citeauthoryear{Kolda, Pinar, Plantenga, and Seshadhri}{Kolda
  et~al\mbox{.}}{2014}]%
        {kolda2014scalable}
\bibfield{author}{\bibinfo{person}{Tamara~G Kolda}, \bibinfo{person}{Ali
  Pinar}, \bibinfo{person}{Todd Plantenga}, {and} \bibinfo{person}{Comandur
  Seshadhri}.} \bibinfo{year}{2014}\natexlab{}.
\newblock \showarticletitle{A scalable generative graph model with community
  structure}.
\newblock \bibinfo{journal}{\emph{SIAM J. Sci. Comput}} \bibinfo{volume}{36},
  \bibinfo{number}{5} (\bibinfo{year}{2014}), \bibinfo{pages}{C424--C452}.
\newblock


\bibitem[\protect\citeauthoryear{Kumar, Raghavan, Rajagopalan, Sivakumar,
  Tomkins, and Uptal}{Kumar et~al\mbox{.}}{2000}]%
        {kumar2000stochastic}
\bibfield{author}{\bibinfo{person}{R. Kumar}, \bibinfo{person}{P. Raghavan},
  \bibinfo{person}{S. Rajagopalan}, \bibinfo{person}{D. Sivakumar},
  \bibinfo{person}{A. Tomkins}, {and} \bibinfo{person}{E. Uptal}.}
  \bibinfo{year}{2000}\natexlab{}.
\newblock \showarticletitle{Stochastic models for the web graph}. In
  \bibinfo{booktitle}{\emph{FOCS}}.
\newblock


\bibitem[\protect\citeauthoryear{Leskovec, Backstrom, Kumar, and
  Tomkins}{Leskovec et~al\mbox{.}}{2008}]%
        {leskovec2008microscopic}
\bibfield{author}{\bibinfo{person}{Jure Leskovec}, \bibinfo{person}{Lars
  Backstrom}, \bibinfo{person}{Ravi Kumar}, {and} \bibinfo{person}{Andrew
  Tomkins}.} \bibinfo{year}{2008}\natexlab{}.
\newblock \showarticletitle{Microscopic evolution of social networks}. In
  \bibinfo{booktitle}{\emph{KDD}}.
\newblock


\bibitem[\protect\citeauthoryear{Leskovec, Chakrabarti, Kleinberg, Faloutsos,
  and Ghahramani}{Leskovec et~al\mbox{.}}{2010}]%
        {leskovec2010kronecker}
\bibfield{author}{\bibinfo{person}{Jure Leskovec}, \bibinfo{person}{Deepayan
  Chakrabarti}, \bibinfo{person}{Jon Kleinberg}, \bibinfo{person}{Christos
  Faloutsos}, {and} \bibinfo{person}{Zoubin Ghahramani}.}
  \bibinfo{year}{2010}\natexlab{}.
\newblock \showarticletitle{Kronecker Graphs: An Approach to Modeling
  Networks}.
\newblock \bibinfo{journal}{\emph{J. Mach. Learn. Res}}  \bibinfo{volume}{11}
  (\bibinfo{year}{2010}), \bibinfo{pages}{985--1042}.
\newblock


\bibitem[\protect\citeauthoryear{Leskovec, Kleinberg, and Faloutsos}{Leskovec
  et~al\mbox{.}}{2005}]%
        {leskovec2005graphs}
\bibfield{author}{\bibinfo{person}{Jure Leskovec}, \bibinfo{person}{Jon
  Kleinberg}, {and} \bibinfo{person}{Christos Faloutsos}.}
  \bibinfo{year}{2005}\natexlab{}.
\newblock \showarticletitle{Graphs over time: densification laws, shrinking
  diameters and possible explanations}. In \bibinfo{booktitle}{\emph{KDD}}.
\newblock


\bibitem[\protect\citeauthoryear{Lilliefors}{Lilliefors}{1969}]%
        {lilliefors1969kolmogorov}
\bibfield{author}{\bibinfo{person}{Hubert~W Lilliefors}.}
  \bibinfo{year}{1969}\natexlab{}.
\newblock \showarticletitle{On the Kolmogorov-Smirnov test for the exponential
  distribution with mean unknown}.
\newblock \bibinfo{journal}{\emph{J. Amer. Statist. Assoc.}}
  \bibinfo{volume}{64} (\bibinfo{year}{1969}), \bibinfo{pages}{387--389}.
\newblock


\bibitem[\protect\citeauthoryear{Liu, Benson, and Charikar}{Liu
  et~al\mbox{.}}{2019}]%
        {liu2018sampling}
\bibfield{author}{\bibinfo{person}{Paul Liu}, \bibinfo{person}{Austin Benson},
  {and} \bibinfo{person}{Moses Charikar}.} \bibinfo{year}{2019}\natexlab{}.
\newblock \showarticletitle{A sampling framework for counting temporal motifs}.
  In \bibinfo{booktitle}{\emph{WSDM}}.
\newblock


\bibitem[\protect\citeauthoryear{Mahadevan, Krioukov, Fall, and
  Vahdat}{Mahadevan et~al\mbox{.}}{2006}]%
        {mahadevan2006systematic}
\bibfield{author}{\bibinfo{person}{Priya Mahadevan}, \bibinfo{person}{Dmitri
  Krioukov}, \bibinfo{person}{Kevin Fall}, {and} \bibinfo{person}{Amin
  Vahdat}.} \bibinfo{year}{2006}\natexlab{}.
\newblock \showarticletitle{Systematic topology analysis and generation using
  degree correlations}. In \bibinfo{booktitle}{\emph{SIGCOMM}}.
\newblock


\bibitem[\protect\citeauthoryear{Milgram}{Milgram}{1967}]%
        {milgram1967small}
\bibfield{author}{\bibinfo{person}{Stanley Milgram}.}
  \bibinfo{year}{1967}\natexlab{}.
\newblock \showarticletitle{The small-world problem}.
\newblock \bibinfo{journal}{\emph{Psychology Today}} \bibinfo{volume}{2},
  \bibinfo{number}{1} (\bibinfo{year}{1967}), \bibinfo{pages}{60--67}.
\newblock


\bibitem[\protect\citeauthoryear{Milo, Shen-Orr, Itzkovitz, Kashtan,
  Chklovskii, and Alon}{Milo et~al\mbox{.}}{2002}]%
        {milo2002network}
\bibfield{author}{\bibinfo{person}{Ron Milo}, \bibinfo{person}{Shai Shen-Orr},
  \bibinfo{person}{Shalev Itzkovitz}, \bibinfo{person}{Nadav Kashtan},
  \bibinfo{person}{Dmitri Chklovskii}, {and} \bibinfo{person}{Uri Alon}.}
  \bibinfo{year}{2002}\natexlab{}.
\newblock \showarticletitle{Network motifs: simple building blocks of complex
  networks}.
\newblock \bibinfo{journal}{\emph{Science}} \bibinfo{volume}{298},
  \bibinfo{number}{5594} (\bibinfo{year}{2002}), \bibinfo{pages}{824--827}.
\newblock


\bibitem[\protect\citeauthoryear{Mitzenmacher}{Mitzenmacher}{2004}]%
        {mitzenmacher2004brief}
\bibfield{author}{\bibinfo{person}{Michael Mitzenmacher}.}
  \bibinfo{year}{2004}\natexlab{}.
\newblock \showarticletitle{A brief history of generative models for power law
  and lognormal distributions}.
\newblock \bibinfo{journal}{\emph{Internet Mathematics}} \bibinfo{volume}{1},
  \bibinfo{number}{2} (\bibinfo{year}{2004}), \bibinfo{pages}{226--251}.
\newblock


\bibitem[\protect\citeauthoryear{Newman}{Newman}{2001}]%
        {newman2001clustering}
\bibfield{author}{\bibinfo{person}{Mark~EJ Newman}.}
  \bibinfo{year}{2001}\natexlab{}.
\newblock \showarticletitle{Clustering and preferential attachment in growing
  networks}.
\newblock \bibinfo{journal}{\emph{Physical review E}} \bibinfo{volume}{64},
  \bibinfo{number}{2} (\bibinfo{year}{2001}), \bibinfo{pages}{025102}.
\newblock


\bibitem[\protect\citeauthoryear{Paranjape, Benson, and Leskovec}{Paranjape
  et~al\mbox{.}}{2017}]%
        {paranjape2017motifs}
\bibfield{author}{\bibinfo{person}{Ashwin Paranjape}, \bibinfo{person}{Austin~R
  Benson}, {and} \bibinfo{person}{Jure Leskovec}.}
  \bibinfo{year}{2017}\natexlab{}.
\newblock \showarticletitle{Motifs in temporal networks}. In
  \bibinfo{booktitle}{\emph{WSDM}}.
\newblock


\bibitem[\protect\citeauthoryear{Sala, Cao, Wilson, Zablit, Zheng, and
  Zhao}{Sala et~al\mbox{.}}{2010}]%
        {sala2010measurement}
\bibfield{author}{\bibinfo{person}{Alessandra Sala}, \bibinfo{person}{Lili
  Cao}, \bibinfo{person}{Christo Wilson}, \bibinfo{person}{Robert Zablit},
  \bibinfo{person}{Haitao Zheng}, {and} \bibinfo{person}{Ben~Y Zhao}.}
  \bibinfo{year}{2010}\natexlab{}.
\newblock \showarticletitle{Measurement-calibrated graph models for social
  network experiments}. In \bibinfo{booktitle}{\emph{WWW}}.
\newblock


\bibitem[\protect\citeauthoryear{Shin}{Shin}{2017}]%
        {shin2017wrs}
\bibfield{author}{\bibinfo{person}{Kijung Shin}.}
  \bibinfo{year}{2017}\natexlab{}.
\newblock \showarticletitle{Wrs: Waiting room sampling for accurate triangle
  counting in real graph streams}. In \bibinfo{booktitle}{\emph{ICDM}}.
\newblock


\bibitem[\protect\citeauthoryear{Shin, Eliassi-Rad, and Faloutsos}{Shin
  et~al\mbox{.}}{2018}]%
        {shin2018patterns}
\bibfield{author}{\bibinfo{person}{Kijung Shin}, \bibinfo{person}{Tina
  Eliassi-Rad}, {and} \bibinfo{person}{Christos Faloutsos}.}
  \bibinfo{year}{2018}\natexlab{}.
\newblock \showarticletitle{Patterns and anomalies in k-cores of real-world
  graphs with applications}.
\newblock \bibinfo{journal}{\emph{Knowl. Inf. Syst}} \bibinfo{volume}{54},
  \bibinfo{number}{3} (\bibinfo{year}{2018}), \bibinfo{pages}{677--710}.
\newblock


\bibitem[\protect\citeauthoryear{Sinha, Shen, Song, Ma, Eide, Hsu, and
  Wang}{Sinha et~al\mbox{.}}{2015}]%
        {sinha2015MAG}
\bibfield{author}{\bibinfo{person}{Arnab Sinha}, \bibinfo{person}{Zhihong
  Shen}, \bibinfo{person}{Yang Song}, \bibinfo{person}{Hao Ma},
  \bibinfo{person}{Darrin Eide}, \bibinfo{person}{Bo-June~(Paul) Hsu}, {and}
  \bibinfo{person}{Kuansan Wang}.} \bibinfo{year}{2015}\natexlab{}.
\newblock \showarticletitle{An Overview of Microsoft Academic Service ({MAS})
  and Applications}. In \bibinfo{booktitle}{\emph{WWW}}.
\newblock


\bibitem[\protect\citeauthoryear{Stasi, Sadeghi, Rinaldo, Petrovi{\'c}, and
  Fienberg}{Stasi et~al\mbox{.}}{2014}]%
        {stasi2014beta}
\bibfield{author}{\bibinfo{person}{Despina Stasi}, \bibinfo{person}{Kayvan
  Sadeghi}, \bibinfo{person}{Alessandro Rinaldo}, \bibinfo{person}{Sonja
  Petrovi{\'c}}, {and} \bibinfo{person}{Stephen~E Fienberg}.}
  \bibinfo{year}{2014}\natexlab{}.
\newblock \showarticletitle{$\beta$ models for random hypergraphs with a given
  degree sequence}.
\newblock \bibinfo{journal}{\emph{arXiv preprint arXiv:1407.1004}}
  (\bibinfo{year}{2014}).
\newblock


\bibitem[\protect\citeauthoryear{V{\'a}zquez}{V{\'a}zquez}{2003}]%
        {vazquez2003growing}
\bibfield{author}{\bibinfo{person}{Alexei V{\'a}zquez}.}
  \bibinfo{year}{2003}\natexlab{}.
\newblock \showarticletitle{Growing network with local rules: Preferential
  attachment, clustering hierarchy, and degree correlations}.
\newblock \bibinfo{journal}{\emph{Physical Review E}} \bibinfo{volume}{67},
  \bibinfo{number}{5} (\bibinfo{year}{2003}), \bibinfo{pages}{056104}.
\newblock


\bibitem[\protect\citeauthoryear{Watts and Strogatz}{Watts and
  Strogatz}{1998}]%
        {watts1998collective}
\bibfield{author}{\bibinfo{person}{Duncan~J Watts} {and}
  \bibinfo{person}{Steven~H Strogatz}.} \bibinfo{year}{1998}\natexlab{}.
\newblock \showarticletitle{Collective dynamics of ‘small-world’networks}.
\newblock \bibinfo{journal}{\emph{Nature}} \bibinfo{volume}{393},
  \bibinfo{number}{6684} (\bibinfo{year}{1998}), \bibinfo{pages}{440}.
\newblock


\bibitem[\protect\citeauthoryear{Xu, Rockmore, and Kleinbaum}{Xu
  et~al\mbox{.}}{2013}]%
        {xu2013hyperlink}
\bibfield{author}{\bibinfo{person}{Ye Xu}, \bibinfo{person}{Dan Rockmore},
  {and} \bibinfo{person}{Adam~M Kleinbaum}.} \bibinfo{year}{2013}\natexlab{}.
\newblock \showarticletitle{Hyperlink prediction in hypernetworks using latent
  social features}. In \bibinfo{booktitle}{\emph{DS}}.
\newblock


\bibitem[\protect\citeauthoryear{Yin, Benson, Leskovec, and Gleich}{Yin
  et~al\mbox{.}}{2017}]%
        {yin2017local}
\bibfield{author}{\bibinfo{person}{Hao Yin}, \bibinfo{person}{Austin~R Benson},
  \bibinfo{person}{Jure Leskovec}, {and} \bibinfo{person}{David~F Gleich}.}
  \bibinfo{year}{2017}\natexlab{}.
\newblock \showarticletitle{Local higher-order graph clustering}. In
  \bibinfo{booktitle}{\emph{KDD}}.
\newblock


\bibitem[\protect\citeauthoryear{Yoon, Song, Shin, and Yi}{Yoon
  et~al\mbox{.}}{2020}]%
        {yoon2020much}
\bibfield{author}{\bibinfo{person}{Se-eun Yoon}, \bibinfo{person}{Hyungseok
  Song}, \bibinfo{person}{Kijung Shin}, {and} \bibinfo{person}{Yung Yi}.}
  \bibinfo{year}{2020}\natexlab{}.
\newblock \showarticletitle{How Much and When Do We Need Higher-order
  Information in Hypergraphs? A Case Study on Hyperedge Prediction}. In
  \bibinfo{booktitle}{\emph{WWW}}.
\newblock


\bibitem[\protect\citeauthoryear{Zhou, Huang, and Scholkopf}{Zhou
  et~al\mbox{.}}{2006}]%
        {zhou2006learning}
\bibfield{author}{\bibinfo{person}{Dengyong Zhou}, \bibinfo{person}{Jiayuan
  Huang}, {and} \bibinfo{person}{Bernhard Scholkopf}.}
  \bibinfo{year}{2006}\natexlab{}.
\newblock \showarticletitle{Learning with Hypergraphs: Clustering,
  Classification, and Embedding}. In \bibinfo{booktitle}{\emph{NIPS}}.
\newblock


\end{thebibliography}
